\documentclass[runningheads]{llncs}

\usepackage[T1]{fontenc}
\def\doi#1{\href{https://doi.org/\detokenize{#1}}{\url{https://doi.org/\detokenize{#1}}}}
\usepackage[colorlinks=true,allcolors=blue]{hyperref}

\usepackage{orcidlink}

\usepackage{amsmath}
\usepackage{amssymb}
\usepackage[misc]{ifsym}
\usepackage{enumitem}

\usepackage{amsfonts}
\usepackage{cleveref}
\Crefformat{figure}{#2Fig.~#1#3}
\Crefmultiformat{figure}{Figs.~#2#1#3}{ and~#2#1#3}{, #2#1#3}{ and~#2#1#3}
\usepackage{tabularx}
\usepackage{stfloats}
\usepackage{booktabs}
\usepackage{ltablex}
\usepackage{capt-of}

\newcommand{\STAB}[1]{\begin{tabular}{@{}c@{}}#1\end{tabular}}
\newcommand{\rot}[1]{\rotatebox[origin=c]{90}{#1}} 
\newcommand{\rotmulti}[2]{\multirow{#1}{*}{\STAB{\rotatebox[origin=c]{90}{#2}}}}

\usepackage{multirow}
\usepackage{wrapfig}

\usepackage{adjustbox}
\usepackage{here}

\usepackage{pgfplots}
\usepackage{caption}
\usepackage{subcaption}

\usepackage{xspace}
\newcommand{\tool}{\textsf{symQV}\xspace}

\pgfplotsset{width=8cm,compat=1.9,height=4cm}

\usepgfplotslibrary{external}
\usepgfplotslibrary{fillbetween}
\usepackage{makecell}

\definecolor{forestgreen}{RGB}{0, 155, 85}

\newcommand{\qmodel}{M_\mathcal{Q}}
\newcommand{\spec}{\varphi}

\newcommand{\hilbert}{\mathcal{H}}

\newcommand{\bigkron}{\bigotimes}

\newcommand{\valuation}{V_0}

\newcommand\defeq{:=}

\newcommand{\spaceLand}{\,\land\,}

\usepackage{graphicx}
\usepackage{tikz}
\usetikzlibrary{quantikz,angles,quotes}

\newsavebox{\singlegateX}
\newsavebox{\singlegateZ}
\newsavebox{\singlegateH}
\newsavebox{\cxgate}
\newsavebox{\czgate}
\newsavebox{\swapgate}
\newsavebox{\measurement}

\newcommand{\singlegate}[1]{{\begin{quantikz}\lstick{} & \gate{#1} & \qw\end{quantikz}}}
\savebox{\singlegateX}{\singlegate{X}}
\savebox{\singlegateZ}{\singlegate{Z}}
\savebox{\singlegateH}{\singlegate{H}}
\savebox{\cxgate}{
    {\begin{quantikz}
        \lstick{} & \ctrl{1} & \qw \\
        \lstick{} & \targ{} & \qw \\
    \end{quantikz}}
}
\savebox{\czgate}{
    {\begin{quantikz}
        \lstick{} & \ctrl{1} & \qw \\
        \lstick{} & \gate{Z} & \qw \\
    \end{quantikz}}
}
\savebox{\swapgate}{
    {\begin{quantikz}
        \lstick{} & \swap{1} & \qw \\
        \lstick{} & \targX{} & \qw \\
    \end{quantikz}}
}

\savebox{\measurement}{
    {\begin{quantikz}
        \lstick{} & \meter{$M$} & \cw
    \end{quantikz}}
}

\newcommand{\complexes}{\mathbb{C}}

\newcommand{\mat}[1] {\begin{bmatrix}#1\end{bmatrix}}

\newcommand{\SWAP}{\mathit{SWAP}}

\usepackage[boxed,vlined,linesnumbered]{algorithm2e}
\let\oldnl\nl
\newcommand{\nonl}{\renewcommand{\nl}{\let\nl\oldnl}}

\DeclareFixedFont{\ttb}{T1}{txtt}{bx}{n}{8} 
\DeclareFixedFont{\ttm}{T1}{txtt}{m}{n}{8}  

\usepackage{listings}
\definecolor{light-gray}{gray}{0.9}

\usepackage{color}
\definecolor{deepblue}{rgb}{0,0,0.5}
\definecolor{deepred}{rgb}{0.6,0,0}
\definecolor{deepgreen}{rgb}{0,0.5,0}

\newcommand\pythonstyle{\lstset{
language=Python,
basicstyle=\linespread{1}\ttm,
morekeywords={self},              
keywordstyle=\ttb\color{deepblue},
emph={MyClass,__init__},          
emphstyle=\ttb\color{deepred},    
stringstyle=\color{deepgreen},
frame=tlbr,                         
showstringspaces=false,
numbers          = left, 
numberstyle      = \scriptsize\color{black},
}}

\newcommand\pythonexternal[2][]{{
\pythonstyle
\lstinputlisting[#1]{#2}}}

\newcommand\pythonstylesmall{\lstset{
language=Python,
basicstyle=\tiny\linespread{1.3}\selectfont\ttm,
morekeywords={self},              
keywordstyle=\ttb\color{deepblue},
emph={MyClass,__init__},          
emphstyle=\ttb\color{deepred},    
stringstyle=\color{deepgreen},
frame=tlbr,                         
showstringspaces=false,
numbers          = left, 
numberstyle      = \scriptsize\color{black},
}}

\newcommand\pythonexternalsmall[2][]{{
\pythonstylesmall
\lstinputlisting[#1]{#2}}}

\newfloat{lstfloat}{htbp}{lop}
\floatname{lstfloat}{Listing}

\Crefname{theorem}{Theorem}{Theorems}
\newtheorem{mydef}{Definition}
\Crefname{mydef}{Definition}{Definitions}
\newtheorem{myenc}{Encoding}
\Crefname{myenc}{Encoding}{Encodings}
\newtheorem{myenc2}{Encoding 2.\!\!}
\Crefname{myenc2}{Encoding 2.\!\!}{Encodings 2.\!\!}
\Crefname{example}{Example}{Examples}

\crefname{lstlisting}{listing}{listings}
\Crefname{lstlisting}{Listing}{Listings}

\newcommand{\sat}{\textsc{Sat}}
\newcommand{\unsat}{\textsc{Unsat}}

\newcommand{\timeout}{\texttt{timeout}}
\newcommand{\memoryerror}{\texttt{out of memory}}

\usepackage[square,numbers]{natbib}
\begin{document}
\title{\textsf{symQV}: Automated Symbolic Verification of Quantum Programs}

%
\author{Fabian Bauer-Marquart\inst{1}\thanks{The work was done while the first author was employed at the University of Konstanz.}%
\textsuperscript{(\Letter)}\orcidlink{0000-0001-9312-1706}\ \and
Stefan Leue\inst{1}\orcidlink{0000-0002-4259-624X} \and 
Christian Schilling\inst{2}\orcidlink{0000-0003-3658-1065}}

\authorrunning{F. Bauer-Marquart et al.}
%
\institute{University of Konstanz, Konstanz, Germany\\
\email{\{fabian.marquart,stefan.leue\}@uni-konstanz.de}\\
\and
Aalborg University, Aalborg, Denmark\\
\email{christianms@cs.aau.dk}
}
\maketitle              

\begin{abstract}
We present \tool, a symbolic execution framework for writing and verifying quantum computations in the quantum circuit model.
\tool can automatically verify that a quantum program complies with a first-order specification.
We formally introduce a symbolic quantum program model.
This allows to encode the verification problem in an SMT formula, which can then be checked with a $\mathbf\delta$-complete decision procedure.
We also propose an abstraction technique to speed up the verification process.
Experimental results show that the abstraction improves \tool's scalability by an order of magnitude to quantum programs with 24 qubits (a $ 2^{24}$-dimensional state space).
\keywords{
Quantum computing \and formal verification \and symbolic execution \and abstraction.
}
\end{abstract}
\section{Introduction}\label{sec:introduction}
Quantum computing 
bears great potential in increasing the scalability of problem solving in many areas such as optimization \cite{kadowaki1998quantum,farhi2014quantum}, database search \cite{DBLP:conf/stoc/Grover96}, cryptography \cite{DBLP:journals/siamcomp/Shor97}, quantum dynamics simulation \cite{DBLP:journals/pnas/ChildsMNRS18}, satisfiability problems \cite{centrone2021experimental}, and machine learning \cite{DBLP:journals/corr/abs-2110-13162}.
Recently, quantum computing has gained momentum with applications in safety-critical domains such as traffic flow \cite{goddard2017will}, aircraft load \cite{DBLP:journals/corr/abs-1903-08189}, logistics \cite{ajagekar2020quantum}, and medical diagnostics \cite{DBLP:journals/corr/abs-2102-06535}.
Furthermore, quantum simulation \cite{Qiskit,DBLP:journals/corr/abs-1803-00652,cirq_developers_2021_5182845} and quantum computers in the cloud \cite{ibm-quantum-roadmap} are now available.

As with classical programs, detecting bugs in quantum programs is a crucial problem.
For classical programs, there exist powerful formal verification techniques to automatically verify that the programs comply with a formal specification \cite{ClarkeHVB2018}.
State-of-the-art verifiers, e.g., for C programs \cite{KLEE,DBLP:conf/cav/BeyerK11,DBLP:conf/tacas/KroeningT14} perform verification \emph{symbolically}:
The developer marks specific program inputs as symbolic so that the verifier knows to use these as the ``search space.'' 
The verifier then proves that all possible inputs to the program comply with the specification.

For quantum programs, this level of automation is not yet available.
In this work, we aim to bridge this gap.
Existing approaches to quantum program analysis can be categorized in three directions:

\textbf{Interactive proof assistants:}
Several approaches \cite{DBLP:journals/corr/abs-1803-00699,DBLP:conf/birthday/LiuWZGLHDY17,DBLP:conf/cav/LiuZWYLLYZ19,DBLP:conf/esop/CharetonBBPV21,DBLP:conf/itp/Hietala0HL021} propose using interactive proof assistants to verify quantum programs.
These works provide a large set of deductions but require familiarity with 
proof assistants such as Coq \cite{Coq} or Isabelle/HOL \cite{Isabelle}, competence in proof-writing, and many hours of manual programming work to conduct the verification.
These techniques are not fully automatic, which would be crucial for keeping pace with the development of quantum algorithms~\cite{quantumalgorithmzoo}.

\textbf{Automated quantum compiler verification:}
Amy \cite{DBLP:journals/corr/abs-1805-06908} proposes an efficient path-sum framework that performs fully automated equivalence checking of a quantum program against a simpler version of the same program, as well as against path-sums that the author uses as specification. The approach is applicable to quantum programs written with quantum gates from the Clifford$+T$ group.
Shi et al.\ \cite{shi2019certiq} use an SMT (satisfiability modulo theories) solver to verify a quantum compiler via equivalence checking. These approaches do not handle general formal specifications.

\textbf{Quantum assertion checking:}
Li et al.\ \cite{DBLP:journals/pacmpl/LiZYDY020} verify assertions during quantum program run-time via projections.
Yu and Palsberg \cite{DBLP:conf/pldi/YuP21} use an abstraction to verify assertions on quantum programs with up to 300 qubits, but the approach is restricted to programs where inputs are fixed to a specific value. This is a severe drawback, as essential quantum algorithms such as teleportation, the quantum Fourier transform \cite{nielsen_chuang_2010}, or Grover's diffusion operator \cite{DBLP:conf/stoc/Grover96} require arbitrarily-valued inputs.

\medskip

In summary, despite the significance of ensuring specification compliance in quantum software engineering, there is still a lack of practical, automated tools for the purpose of symbolic quantum verification of general formal specifications.
Existing tools either:
\begin{itemize}
    \item require a high amount of manual programming, 
    \item restrict the type of quantum program, e.g., support only a subset of quantum gates or only measurement-free quantum programs,
    \item do not work symbolically, requiring to fix the inputs to the program, or
    \item do not support the checking of formal specifications written in first-order logic, which is the standard for classical software verification.
\end{itemize}

\medskip

In this paper, we introduce \tool, a framework for writing and verifying quantum programs in the quantum circuit model. 
To the best of our knowledge, \tool is the first tool that allows automated ``push-button'' verification of quantum programs where the programs are executed symbolically.
In \emph{symbolic execution}, a program is not executed with a predetermined input value. Instead, it is executed with the complete range of possible input values.
In contrast to the classical case, where the number of possible input values is bounded by the RAM architecture, the range of input values to a quantum program is infinite.

\tool's automation and high-level workflow are similar to classical verification frameworks such as CPAchecker \cite{DBLP:conf/cav/BeyerK11}:
quantum developers only need to write a quantum program (using a Cirq-like \cite{cirq_developers_2021_5182845} syntax) and a first-order logic specification that expresses the desired program output. 
Then, compliance with this specification is automatically verified based on SMT technology.
If the quantum program does not satisfy the specification, the user obtains a counterexample that aids in locating errors in the program.

A major obstacle in practice is that quantum program simulators require exponential memory in the number of qubits.
This is because simulators running on classical computers need to utilize a matrix to represent the state of a quantum mechanical system.
This matrix doubles in size with every qubit that is added to the computation \cite{nielsen_chuang_2010}, which naturally carries over to verifying quantum programs.
We show that in many practical cases this exponential matrix representation can be avoided.
In addition, we propose an \emph{abstraction} (or \emph{over-approximation}) \cite{CousotC77} that makes our technique more scalable without harming verification soundness.

We evaluate our approach \tool on essential quantum algorithms and subroutines. These include teleportation, QFT, \cite{nielsen_chuang_2010}, Grover's diffusion operator \cite{DBLP:conf/stoc/Grover96}, and quantum phase estimation \cite{DBLP:journals/siamcomp/Shor97}. We demonstrate that \tool efficiently verifies quantum programs with up to 24 symbolic input qubits (a $2^{24}$-dimensional state space), showing its potential to be used as a general-purpose verifier by developers of quantum programs. To put this number into perspective: state-of-the-art quantum computers currently offer one error-corrected qubit \cite{krinner2022realizing}.

The main contributions of this paper can be summarized as follows.
\textbf{First}, we introduce a symbolic quantum program model to express quantum programs and safety specifications in our verification framework. \textbf{Second}, we provide an encoding of the quantum program model in SMT and show that this encoding is sound and complete. We use this encoding to automatically verify formal specifications written in first-order logic. 
\textbf{Third}, we introduce a sound abstraction technique, which improves the verification time by one order of magnitude. 
\textbf{Finally}, we evaluate our implementation \tool on several quantum programs with up to 24 qubits.

\section{Background}\label{sec:background}

This section briefly introduces the concepts of quantum computing used in this paper. For detailed explanations, we refer to Nielsen and Chuang \cite{nielsen_chuang_2010}.

The \textit{qubit} is the basic unit of quantum information.
A single qubit can be in the \emph{ground state} $\ket{0}$ (``ket zero'') or in the \emph{excited state} $\ket{1}$ (``ket one''). 
In general, however, a qubit is in a superposition of both \emph{computational basis states}, written as $\ket{q} = \alpha \ket{0} + \beta \ket{1}$.
The \emph{amplitudes} $\alpha, \beta \in \complexes$ characterize a qubit, with $|\alpha|^2$ and $|\beta|^2$ being the probability of the qubit to be in either state. Therefore, their values are restricted such that $|\alpha|^2 + |\beta|^2 = 1$.
Qubits are often written as two-dimensional vectors:
\begin{align*}
    \ket{0} \equiv \mat{1 \\ 0}, \qquad \ket{1} \equiv \mat{0 \\ 1}, \qquad \ket{q} \equiv \mat{\alpha \\ \beta}.
\end{align*}

The qubit states span a two-dimensional Hilbert space $\hilbert_2 = \{\alpha \ket{0} + \beta\ket{1}\}$, a complete complex vector space where the inner product is defined.
When we combine $n$ qubits, the system's state vector $\ket{\psi}$ spans the \emph{tensor product} of Hilbert spaces $\hilbert_{2^n} = \bigkron_{i=1}^n \hilbert_2^{(i)}$, and $\ket{\psi}$ is a $2^n$-dimensional vector.
\smallskip

\textit{Quantum logic gates} are the building blocks of quantum programs and transform a quantum state into a new quantum state. They are characterized by unitary matrices $U$ that transform quantum state vectors. %
Common quantum gates, shown in \Cref{fig:example-gates}, include $X$ (\textsc{Not}), $Z$ (phase-flip), $H$ (Hadamard), $U_{CX}$ (controlled-\textsc{Not}), and $U_{CZ}$ (controlled phase-flip). 
\begin{figure}[H]
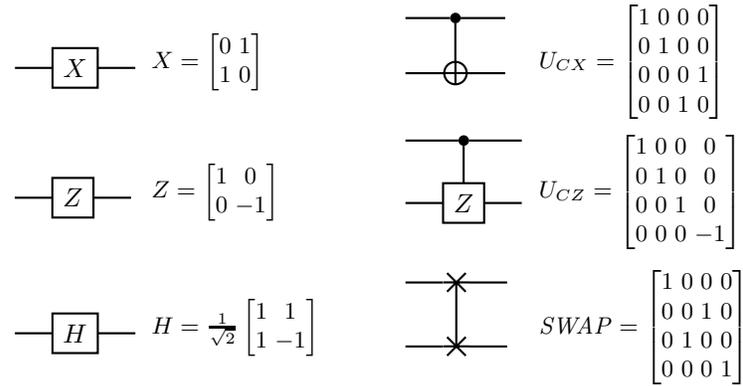

    \centering
    \vspace{-1em}
    \begin{tabular}{ll@{\qquad}ll}
        \usebox\singlegateX
        &
        $X = \mat{0 & 1 \\ 1 & 0}$ 
        &
        \usebox\cxgate
        &
        $U_{CX} = \mat{
            1 & 0 & 0 & 0 \\
            0 & 1 & 0 & 0 \\
            0 & 0 & 0 & 1 \\
            0 & 0 & 1 & 0}$
        \\
        \usebox\singlegateZ
        &
        $Z = \mat{1 & 0 \\ 0 & -1}$
        &
        \usebox\czgate
        &
        $U_{CZ} = \mat{
            1 & 0 & 0 & 0 \\
            0 & 1 & 0 & 0 \\
            0 & 0 & 1 & 0 \\ 
            0 & 0 & 0 & -1}$
        \\
        \usebox\singlegateH 
        &
        $H = \frac{1}{\sqrt{2}} \mat{1 & 1 \\ 1 & -1}$
        &
        \usebox\swapgate
        &
        $\mathit{SWAP} = \mat{
            1 & 0 & 0 & 0 \\
            0 & 0 & 1 & 0 \\
            0 & 1 & 0 & 0 \\ 
            0 & 0 & 0 & 1}$
    \end{tabular}
    \caption{Circuit diagrams and matrices of some common quantum gates. For the controlled gates $U_{CX}$ and $U_{CZ}$, the dot ($\bullet$) marks the control qubit.}
    \label{fig:example-gates}
\end{figure}
\begin{wrapfigure}[13]{r}{0.43\textwidth}
    \vspace{-2.6em}
    \centering
	\scalebox{0.87}{
		\begin{tikzpicture}[scale=1.0,a/.style={->,thick,>=stealth}]
		\def\r{2}
		\def\rr{2.4}
		
		\draw[thick] (\r/2,\r/3) node[circle,fill,inner sep=1,label=above:$\Large{\ket{q}}$] (psi) {};
		\draw (0,0) node[circle,fill,inner sep=1] (orig) {} -- (psi);
		\draw[dashed,thick] (orig) -- (\r/2,-\r/5) node (phi) {} -- (psi);
		
		\draw[thick] (orig) circle (\r);
		\draw[dashed,thick] (orig) ellipse (\r{} and \r/3);
		
		\draw[a,red] (orig) -- ++(-\rr/5,-\rr/3) node[below right] (x) {$x$};
		\draw (-\r/5,-\r/3) node[circle,fill,inner sep=1,label=below left:$\ket{+}$] (+) {};
		\draw (\r/5,\r/3) node[circle,fill,inner sep=1,label=above:$\ket{-}$] (-) {};
		
		\draw[a,blue] (orig) -- ++(\rr,0) node[right] (y) {$y$};
		\draw (\r,0) node[circle,fill,inner sep=1,label=above right:$\ket{i}$] (i) {};
		\draw (-\r,0) node[circle,fill,inner sep=1,label=left:$\ket{-i}$] (-i) {};
		
		\draw[a,forestgreen] (orig) -- ++(0,\rr) node[above] (z) {$z$};
		\draw (0,\r) node[circle,fill,inner sep=1,label=above right:$\ket{0}$] (0) {};
		\draw (0,-\r) node[circle,fill,inner sep=1,label=below:$\ket{1}$] (1) {};
		
		\pic [draw=gray,text=gray,->,"$\phi$"] {angle = x--orig--phi};
		\pic [draw=gray,text=gray,<-,"$\theta$"] {angle = psi--orig--z};
		\end{tikzpicture}}
	\vspace{-1.3em}
	\caption{A qubit $\ket{q}$ visualized on the Bloch sphere.}
	\label{fig:qubit}
\end{wrapfigure}
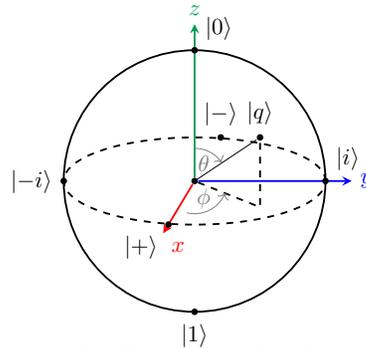
The state of a qubit can alternatively be described with polar coordinates,
$$\ket{q} = \cos \frac{\theta}{2} \ket{0} + e^{i \phi} \sin \frac{\theta}{2} \ket{1},$$
where $\phi$ and $\theta$ correspond to angles that describe a point on the unit sphere, known as the \textit{Bloch sphere} (see \Cref{fig:qubit}), with $\ket{0}$ being the north pole and $\ket{1}$ being the south pole.
For instance, the gates $X$ and $Z$ perform a 180° rotation around the $x$ and $z$ axes, respectively, while $H$ maps ground state $\ket{0}$ to $\ket{+} = \frac{1}{\sqrt{2}}(\ket{0} + \ket{1})$ at the equator.

\subsection{Entanglement}\label{subsec:entanglement}

\emph{Quantum entanglement} is an important concept of quantum mechanics. It occurs if the state of one qubit cannot be characterized independently of the state of another qubit, including when the qubits are separated over a large distance.
Two-qubit states with perfect correlation are called the \textit{Bell states}. An example for such a state is $
    \ket{\phi^+} = \frac{1}{\sqrt{2}}\left(\ket{0} \otimes \ket{0} + \ket{1} \otimes \ket{1}\right),
$ where the first and second qubit are always guaranteed to be either both 0 or both 1 after measurement.

\subsection{Quantum Measurement}\label{subsec:measurement}

Measuring a single qubit $\ket{\psi} = \alpha \ket{0} + \beta \ket{1}$ converts it into a classical bit: $0$ with probability $|\alpha|^2$ and $1$ with probability $|\beta|^2$.
In circuit notation, a measurement is denoted as
\hspace*{-4mm}\begin{quantikz}[baseline=-2mm]\lstick{} & \meter{$M$} & \cw \end{quantikz} (the double stroke indicates a \emph{classical} wire).
Because there are two statistical outcomes, 0 and 1, there exists one measurement operator (a non-unitary matrix) for each: $
    M_0 = \mat{1 & 0 \\ 0 & 0} \text{ and } M_1 = \mat{0 & 0 \\ 0 & 1}
$. 
The measurement operators irreversibly change the quantum state, which influences subsequent computations.
Because of the statistical nature of quantum measurement, simulation tools (and also \tool) need to branch out into two execution paths, with a probability value associated with each of the paths.

\subsection{Running Example: Teleportation}\label{subsec:tp}
Quantum teleportation (TP) is an example of a quantum program with \emph{symbolic} inputs; here, Alice wants to send a qubit $\ket{\psi}$ to Bob.
There exists no quantum communication channel in this problem setting, but Alice and Bob each have one qubit of an entangled qubit pair $\ket{\phi^+}$.
This is used to send (teleport) Alice's qubit to Bob: First, Alice uses a \textit{CNOT} and \textit{H} gate to entangle her two qubits with each other. Then, after measuring both, she sends the measurement results via a classical communication channel to Bob, who finally retrieves $\ket{\psi}$ using two controlled gates, $U_{CX}$ and $U_{CZ}$.
The circuit diagram is shown in \Cref{fig:tp}.

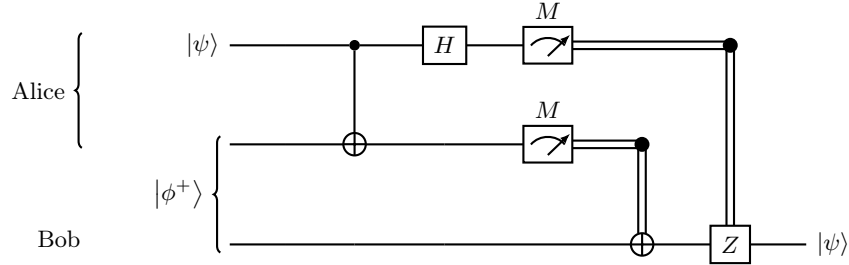
\begin{figure}[h]
    \centering
    \scalebox{1.0}{
    \begin{quantikz}
    \lstick{} \\[-1em]
    \lstick[wires=2]{Alice}
    & \\[2.5em]
    & \\[2.5em]
    \lstick{Bob}
    &  
    \end{quantikz}
    %
    \begin{quantikz}[column sep=2.3em]
    \lstick{$\ket{\psi}$}               
    & \qw 
    & \ctrl{1} 
    & \gate{H}  
    & \meter{$M$} 
    & \cw      
    & \cwbend{2} \\
    \lstick[wires=2]{$\ket{\phi^+}$}
    & \qw                   
    & \targ{}                           
    & \qw                               
    & \meter{$M$}                         
    & \cwbend{1}    
    & \\[1em]
    \qw                                 
    & \qw                   
    & \qw                               
    & \qw                               
    & \qw                                   
    & \targ{} & \gate{Z}    
    & \qw \rstick{\ket{\psi}}
    \end{quantikz}
    }
    \caption{Quantum teleportation circuit, adapted from \cite{nielsen_chuang_2010}. The double line indicates a classical wire. Here, it simulates a communication channel.}
    \label{fig:tp}
\end{figure}

This example motivates the importance of \emph{symbolic verification}: we want to verify that teleportation is successful for \emph{any} quantum state and, hence, need to represent the input state symbolically.

\section{The \tool Quantum Program Model}\label{sec:modeling-quantum-programs}

We introduce the \emph{quantum program model} $\qmodel$ as an SMT-compatible symbolic representation of the general quantum circuit model \cite{nielsen_chuang_2010}.
The quantum program model, unlike the standard state-vector representation used in simulators, can represent operations on qubits as direct mappings in SMT instead of matrices.
Only when necessary, for example when qubits become entangled, do we construct the state vector for this specific subset of qubits.

The main benefit of the quantum program model is that it allows reasoning about quantum programs whose inputs are symbolic and therefore not fixed to a certain value. 
Thus we can use the model to perform formal verification against all possible inputs, i.e., the entire infinite Hilbert space.
Furthermore, the quantum program model allows us to handle quantum programs with parametrized gates, which add another (infinite) dimension to the problem.
\smallskip

We give a high-level, bottom-up presentation of the quantum program model.
At the end of the presentation we exemplify the encoding of the quantum teleportation program in \Cref{sec:example_encoding_teleportation} (the complete SMT formula is shown in \Cref{sec:complete_smt_encoding_teleportation}).
First, we need symbolic encodings for qubits, computations, and measurements.
For convenience, we encode both the amplitudes and the phases into the qubit's SMT representation, allowing computations to work on either.

\begin{myenc}[Qubit]\label{enc:modeling-of-qubits}
    We encode a complex number as a pair $z \defeq (z_R, z_I)$ with $z_R, z_I \in \mathbb{R}$.
    Using this representation, we encode a qubit as a 4-tuple\footnote{We choose $\alpha$ to be real because the global phase \cite{nielsen_chuang_2010} has no observable consequences.}
    \begin{align*}
        \ket{q} &\defeq (\alpha, \beta, \phi, \theta), \qquad \alpha, \phi, \theta \in \mathbb{R}, \ \beta \in \complexes.
    \end{align*}
    We combine both the amplitude and phase representation because we need to restrict the valuations of the variables using the following constraints:
    \begin{align}
        \alpha &= \cos \frac{\theta}{2} \ \land\ 
        \beta_R = \cos \phi \cdot \sin \frac{\theta}{2} \ \land\
        \beta_I = \sin \phi \cdot \sin \frac{\theta}{2}, \label{eq:qubit-degrees-of-freedom}
    \end{align}
    which constrains the qubit's degrees of freedom to $|\alpha|^2 + |\beta|^2 = 1$, and
    \begin{align}
        0 \leq \theta \leq \pi \spaceLand 0 \leq \phi < 2\pi  \quad 
        \land \quad \theta = 0 \Rightarrow \phi = 0  \quad 
        \land \quad \theta = \pi \Rightarrow \phi = 0, \label{eq:qubit-periods}
    \end{align}
    which constrains the angles' values to their respective periods.
\end{myenc}

\Cref{enc:modeling-of-qubits} constrains a qubit's degree of freedom via its phases (\Cref{eq:qubit-periods}). This is because directly encoding the sphere equation $|\alpha|^2 + |\beta|^2 = 1$ requires two nested square operations, which are challenging for state-of-the-art SMT solvers (we evaluated Z3~\cite{DBLP:conf/tacas/MouraB08} and dReal~\cite{DBLP:conf/cade/GaoKC13}).

The main motivation for our quantum program model is that we are often not required to build the whole ($2^n$-dimensional) state vector.
Standard (unitary) quantum gates can be conveniently realized by a direct mapping on the SMT level, which we first define in an abstract way and instantiate later:

\begin{mydef}[Direct mapping]\label{enc:mapping}
    We encode a unitary gate as a bijection $U: \hilbert_2^k \to \hilbert_2^k$ called \emph{direct mapping}, where $k$ is the number of modified qubits.
\end{mydef}

Direct mappings allow us to express the effect of a quantum gate without explicitly constructing the matrix representation, unlike in standard quantum simulators.
We concretize the notion of the direct mapping (\Cref{enc:mapping}) with the following encodings of the most common quantum logic gates \cite{nielsen_chuang_2010}:

\begin{myenc2}[Basic single-qubit gates]\label{enc:basic-single-qubit-gates}
    The identity, $X$, $Z$, and $H$ gates are encoded as the following mappings:
    \begin{align*}
        I \left( \mat{\alpha \\ \beta} \right) \defeq \mat{\alpha \\ \beta}\!, \, 
        X \left( \mat{\alpha \\ \beta} \right) \defeq \mat{\beta \\ \alpha}\!, \, 
        Z \left( \mat{\alpha \\ \beta} \right) \defeq \mat{\alpha \\ -\beta}\!, \,
        H \left( \mat{\alpha \\ \beta} \right) \defeq \mat{\frac{\alpha + \beta}{\sqrt{2}} \\  \frac{\alpha - \beta}{\sqrt{2}}}\!. 
    \end{align*}
    
    We extend the encoding of the identity gate to take a variable number of arguments, such that $I(\ket{q_0}, \dots, \ket{q_k}) = (\ket{q_0}, \dots, \ket{q_k})$ for any $k$.
\end{myenc2}

The gates in \Cref{enc:basic-single-qubit-gates} are used to modify the amplitudes of a qubit. The next encoding includes gates that modify a qubit's phases without directly affecting its amplitudes.

\begin{myenc2}[Phase gates]
    The phase gates $R_X$ and $R_Z$ perform parame\-trized rotations around the $x$ and $z$ axes, respectively. The mappings use the phase angles:
    \begin{align*}
        R_X(\theta')(\phi, \theta) \defeq (\phi, \theta + \theta'), \qquad R_Z(\phi')(\phi, \theta) \defeq (\phi + \phi', \theta).
    \end{align*}
\end{myenc2}

\begin{myenc2}[\textit{SWAP} gate]\label{enc:swap-gate}
The mapping of the $\SWAP$ gate applied to qubits $\ket{q_0}$ and $\ket{q_1}$ is
    \begin{align*}
        \SWAP(\ket{q_0}\!, \ket{q_1}) \defeq (\ket{q_1}, \ket{q_0}).
    \end{align*}
\end{myenc2}

In cases where it is not possible to express a quantum gate as a unitary mapping, such as entangling gates, we resort to the standard matrix representation.
The matrix is then applied to a quantum state vector via matrix multiplication.

\addtocounter{myenc}{1}
\begin{myenc}[Gate matrix]\label{enc:gate-matrix}
    We encode a quantum gate as a $2^k \times 2^k$ (complex) matrix $U$, where $k$ is the number of modified qubits.
    We further require that $U$ is reversible (cf.\ \Cref{sec:background}).
\end{myenc}

\begin{myenc}[Matrix multiplication]\label{enc:matrix-multiplication}
    For an $m \times n$ matrix $A$ and an $n \times p$ matrix $B$, the result of the matrix multiplication $A \cdot B$, an $m \times p$ matrix $C$, is encoded via the identities $\bigwedge_{i = 1}^m \bigwedge_{j = 1}^p c_{i,j} = \sum_{k=1}^n a_{i,k} b_{k,j}$.
\end{myenc}

There are benefits when encoding a gate via a direct mapping instead of a matrix, which we now illustrate with an example: 

\begin{example}\label{ex:swap}
    Recall that the \textit{SWAP} gate can be encoded via a direct mapping (\Cref{enc:swap-gate}), i.e., we can compute 
    \begin{align*}
        \SWAP(\ket{q_0}\!, \ket{q_1}) = (\ket{q_1}\!, \ket{q_0})
    \end{align*}
    in one step. This is \emph{not} the case for the matrix encoding:
    \begin{align*}
        \SWAP (\ket{q_0} \otimes \ket{q_1})
        &= \mat{
            1 & 0 & 0 & 0 \\
            0 & 0 & 1 & 0 \\
            0 & 1 & 0 & 0 \\ 
            0 & 0 & 0 & 1} \mat{\alpha_0 \\ \beta_0} \otimes \mat{\alpha_1 \\ \beta_1} 
        = \mat{
            1 & 0 & 0 & 0 \\
            0 & 0 & 1 & 0 \\
            0 & 1 & 0 & 0 \\ 
            0 & 0 & 0 & 1} \mat{\alpha_0 \alpha_1 \\ \alpha_0 \beta_1 \\ \beta_0 \alpha_1 \\ \beta_0 \beta_1} 
        = \mat{\alpha_0 \alpha_1  \\ \beta_0 \alpha_1 \\ \alpha_0 \beta_1 \\ \beta_0 \beta_1} \\[2mm]
        &= \ket{q_1} \otimes \ket{q_0}.
    \end{align*}
    Here we observe that the matrix representation is verbose. It 
    needs 4 multiplications per tensor product and 16 multiplications only for computing the result of the matrix multiplication. Note that the number of operations increases exponentially with the number of qubits, illustrating the benefit of the direct mapping.
    We give a further example of a direct mapping in Appendix~\ref{app:example}.
\end{example}

Measurement, the only non-reversible operation in our encodings, assigns $0$ or $1$ to a qubit with a certain probability.
For a state $s$ consisting of a single qubit $\ket{q} = \alpha \ket{0} + \beta \ket{1}$, there are two possible subsequent states: $s'(0) = \ket{0}$ and $s'(1) = \ket{1}$. The probabilities $p(x)$ that state $x$ occurs are
\begin{align*}
    p(0) = |\alpha|^2,\ p(1) = |\beta|^2.
\end{align*}
Therefore, for every quantum measurement taking place in $\qmodel$, in the case of non-zero probabilities $p(0)$ and $p(1)$, there are two possible successor states, one per measurement outcome.

\begin{myenc}[Quantum measurement]\label{enc:quantum-measurement}
    We encode the measurement operators by applying the standard measurement matrices (cf.\  \Cref{sec:background}) to \Crefrange{enc:gate-matrix}{enc:matrix-multiplication}.
\end{myenc}

For entangled quantum states, qubits can no longer be characterized individually \cite{nielsen_chuang_2010}. Therefore, our encoding cannot use the direct-mapping strategy from \Cref{enc:mapping} and we fall back to a vector representation of the quantum state.

\begin{mydef}[Modeling a quantum state]\label{enc:modeling-a-quantum-state}
    We define a vector data structure to represent an $n$-qubit quantum state $\ket{\psi}$. 
    This structure holds (cf.\ \Cref{sec:background}) $2^n$ (symbolic) complex numbers
    \begin{align*}
        \ket{\psi} &\defeq (\alpha_1, \alpha_2, \cdots, \alpha_{2^n}).
    \end{align*}
\end{mydef}

\begin{myenc}[Tensor product of matrices]\label{enc:tensor-product}
    For an $m \times n$ matrix $A$ and a $p \times q$ matrix $B$, the tensor product $A \otimes B$, an $(mp) \times (nq)$ matrix $C$, is encoded via equalities $\bigwedge_{i=1}^{m} \bigwedge_{k=1}^{p} \bigwedge_{j=1}^{n} \bigwedge_{l=1}^{q} c_{ik, jl} = a_{i,j} \cdot b_{k,l}$.
\end{myenc}

The following encoding is needed for gate matrices that only apply to a subset of the qubits in the system. This is achieved by taking a tensor product with the identity matrix $I$.

\begin{myenc}[Applying gates to a subset of qubits]\label{enc:gates-subset-of-qubits}
    For a quantum state $\ket{\psi}$ over $n + 1$ qubits and a quantum gate $U$ over qubits $\ket{q_i}$ to $\ket{q_j}$ where $0 \leq i < j \leq n$,
    the next state is 
    \begin{align*}
        \ket{\psi'} &= \begin{cases}
        I^{\otimes i - 1} \otimes U \otimes I^{\otimes n-j} \ket{\psi}  & \mathrm{if} \ 0 < i, j < n, \\
        U \otimes I^{\otimes n-j} \ket{\psi}                            & \mathrm{if} \ 0 = i, j < n, \\
        I^{\otimes i - 1} \otimes U \ket{\psi}                          & \mathrm{if} \ 0 < i, j = n, \\
        U \ket{\psi}                                                    & \mathrm{if} \ 0 = i, j = n.
        \end{cases}
    \end{align*}
\end{myenc}

Having assigned a logic representation to qubits, quantum gates, and quantum measurement, we can combine them to define the \emph{quantum program model}.

\begin{mydef}[Quantum program model]\label{def:quantum-program-model}
    A \emph{quantum program model} is a 5-tuple
    \begin{align}
        \qmodel \defeq (\mathcal{Q}, S, \to, \Theta, \valuation)
    \end{align}
    where 
    \begin{itemize}
        \item $\mathcal{Q}$ is a set of $n$ (symbolic) qubits $\{\ket{q_0}, \dots, \ket{q_{n-1}}\}$,
        \item $S$ is a sequence of $m$ (symbolic) states $(s_{0}, \dots, s_{m-1})$,
        \item $\to$ is a sequence of $m-1$ state operations $(\to_1, \linebreak[1] \dots, \to_{m-1})$,
        \item $\Theta$ is a set of (symbolic) parameters, and
        \item $\valuation$ is the qubit initializer sequence.
    \end{itemize}
\end{mydef}

The qubits of $\mathcal{Q}$ are symbolic unless an initial valuation (assignment of a subset of qubits with concrete values) is provided in $\valuation$.
The initial state is $s_0 = (\ket{q_{0,0}}, \dots, \ket{q_{0,n-1}})$ and all following states
$s_i \in S$ ($0 < i < m$) again consist of symbolic qubits $(\ket{q_{i,0}}, \dots, \ket{q_{i,n-1}})$.
Every state operation $\to_i$ is either
\begin{itemize}
    \item a direct mapping (\Cref{enc:mapping}); or
    \item a unitary matrix (\Cref{enc:gate-matrix}); or 
    \item a quantum measurement (\Cref{enc:quantum-measurement}).
\end{itemize}

We define the shorthand
\begin{align*}
    &s_{i-1} \to_i s_i =\begin{cases}
    \to_i(s_{i-1}) = s_i & \quad  \to_i \text{is a direct mapping}, \\
    \to_i \cdot \bigkron_{j = 0}^{n-1} \ket{q_{(i-1,j)}} = \bigkron_{j = 0}^{n-1} \ket{q_{(i,j)}} & \quad \to_i \text{is a matrix},
    \end{cases}
\end{align*}

and tie the states and operations together via $
    \bigwedge_{i=1}^{m-1} s_{i-1} \to_i s_i.
$

A state operation can also be a quantum measurement $M$.
When state $s_{i-1}$ is measured, two possible subsequent states are created: $s_{i}(0)$ and $s_{i}(1)$ (\Cref{subsec:measurement}).
Additionally, we allow measurement of $k$ qubits at the same time
for a bit vector $x \in \{0, 1\}^k$ such that $M_x$ is the combined measurement.

The set $\Theta$ contains symbolic, real-valued variables that are used to parameterize state operations, e.g., rotations.
The sequence
$
    \valuation = \left( \Psi_0, \dots, \Psi_{n-1} \right)
$
contains sets of initial valuations $\Psi_i \subseteq \hilbert_2$ (possibly  singleton sets in case of a concrete valuation).
The initial valuations are asserted to the initial qubits via
$
    \bigwedge_{i = 0}^{n-1} \ket{q_i} \in \Psi_i.
$

\medskip 

Before we give an example, we note that the quantum program model $\qmodel$ is equivalent to the traditional presentation of quantum computing.

\begin{theorem}[Equivalence]\label{th:equivalence}
    The quantum program model $\qmodel$ (\Cref{def:quantum-program-model}) and the quantum circuit model \cite{nielsen_chuang_2010} are equivalent.
\end{theorem}

The proof for \Cref{th:equivalence} is given in Appendix~\ref{app:proof}.

\subsection{Running Example: Quantum Program Model of Teleportation}\label{sec:example_encoding_teleportation}
Now that we have defined the quantum program model, we formalize our running example, teleportation, as
$\qmodel = (\mathcal{Q}, S, \to, \varnothing, \valuation)$, where
\begin{align*}
    \mathcal{Q} &= \{ \ket{q_0}\!, \ket{q_1}\!, \ket{q_2} \}, \\
    S &= ( s_0, s_1, s_2, s_3, s_4 ), \\
    \to &= (U_{C\!X}(\ket{q_0}\!\!, \ket{q_1}), 
    H(\ket{q_0}), 
    \mathcal{M}(\ket{q_0}\!\!, \ket{q_1}), 
    U_{C\!X}(\ket{q_1}\!\!, \ket{q_2}), 
    U_{C\!Z}(\ket{q_0}\!\!, \ket{q_2}), \\
    \valuation &= (\hilbert_2, \{\ket{\phi^+}\}).
\end{align*}

Note that valuations $\valuation$ are symbolic, so each input qubit can assume any state in the Hilbert space.

\medskip

Next we provide a high-level encoding of this quantum program model in SMT.
The complete SMT formula is shown in \Cref{sec:complete_smt_encoding_teleportation}.

We begin by encoding the first state $s_0$, which contains the three input qubits $\ket{q_{0,0}}\!, \ket{q_{0,1}}\!, \ket{q_{0,2}}$. The first operation $s_0 \to_1 s_1$ is encoded as $\ket{q_{1,0}} \otimes \ket{q_{1,1}} = U_{CX} \ket{q_{0,0}} \otimes \ket{q_{0,1}}$, with $s_1$ containing the qubits $\ket{q_{1,0}}\!, \ket{q_{1,1}}\!, \ket{q_{1,2}}$ that encode the result of this operation.
The remaining states and state operations are encoded as follows (we have omitted identity operations for the sake of brevity), with all entries connected with a conjunction:
\begin{center}
\small
\begin{tabularx}{1.1\linewidth}{@{}lX@{}} \toprule
    \textbf{State} & \textbf{Operation} \\  \midrule
    $s_2 = (\ket{q_{2,0}}\!, \ket{q_{2,1}}\!, \ket{q_{2,2}})$ & $\ket{q_{2,0}} = H \ket{q_{1,0}}$ \\
    $s_3(00) = (\ket{q_{3,0}(00)}\!, \ket{q_{3,1}(00)}\!, \ket{q_{3,2}(00)})$
    & $\ket{q_{3,0}(00)} = M_0 \ket{q_{2,0}}\!, \ \ket{q_{3,1}(00)} = M_0 \ket{q_{2,1}}$ \\
    $s_3(01) = (\ket{q_{3,0}(01)}\!, \ket{q_{3,1}(01)}\!, \ket{q_{3,2}(01)})$
    & $\ket{q_{3,0}(01)} = M_0 \ket{q_{2,0}}\!, \ \ket{q_{3,1}(01)} = M_1 \ket{q_{2,1}}$ \\
    $s_3(10) = (\ket{q_{3,0}(10)}\!, \ket{q_{3,1}(10)}\!, \ket{q_{3,2}(10)})$
    & $\ket{q_{3,0}(10)} = M_1 \ket{q_{2,0}}\!, \ \ket{q_{3,1}(10)} = M_0 \ket{q_{2,1}}$ \\
    $s_3(11) = (\ket{q_{3,0}(11)}\!, \ket{q_{3,1}(11)}\!, \ket{q_{3,2}(11)})$ 
    & $\ket{q_{3,0}(11)} = M_1 \ket{q_{2,0}}\!, \ \ket{q_{3,1}(11)} = M_1 \ket{q_{2,1}}$ \\
    $s_4(x) = (\ket{q_{4,0}(x)}\!, \ket{q_{4,1}(x)}\!, \ket{q_{4,2}(x)})$ 
    & $\ket{q_{4,1}(x)} \otimes \ket{q_{4,2}(x)} = U_{CX} \ket{q_{3,1}(x)} \otimes \ket{q_{3,2}(x)}$ \\
    \hspace{1em}$(x \in \{00, 01, 10, 11\})$ \\
    $s_5(x) = (\ket{q_{5,0}(x)}\!, \ket{q_{5,1}(x)}\!, \ket{q_{5,2}(x)})$
    & $\ket{q_{5,0}(x)} \otimes \ket{q_{5,2}(x)} = U_{CZ} \ket{q_{4,0}(x)} \otimes \ket{q_{4,2}(x)}$ \\
    \midrule
    \textbf{Initial valuation} & $\ket{q_{0,0}} \in \hilbert_2, \ \ket{q_{0,1}} \otimes \ket{q_{0,2}} \in \{ \ket{\phi^+} \}$ \\
    \bottomrule
\end{tabularx}
\end{center}
We observe that the measurement step from $s_2$ to $s_3$ results in the creation of 4 possible execution paths, one per measurement outcome ($00, 01, 10, 11$).
Also, recall that all the symbols and operators used in the encoding above, such as the tensor product ($\otimes$), gates ($H$, $U_{CX}$, $U_{CZ}$), measurements ($M_0$, $M_1$), and Hilbert space ($\hilbert_2$), carry the meanings we assigned to them in  \Crefrange{enc:modeling-of-qubits}{enc:gates-subset-of-qubits}.

\section{The \tool Verification Algorithm}

Our \tool algorithm takes as input a quantum program model $\qmodel$ defined in \Cref{sec:modeling-quantum-programs} and a formal specification in the form of a first-order formula $\spec$.
From that, \tool generates an SMT encoding (which we also write $\qmodel$ with a slight abuse of notation) as described in the previous section.
Finally, this encoding together with the negated specification is asserted in a query to an SMT solver.

\begin{theorem}[Soundness and completeness of the encoding]\label{sound-complete}
    Given a quantum program model with encoding $\qmodel$ and a specification $\spec$, we have that the program satisfies $\spec$ if and only if $\qmodel \land \neg \spec$ is unsatisfiable.
\end{theorem}

\begin{proof}
    This follows from the one-to-one correspondence of the quantum program model $\qmodel$ and the standard quantum circuit model \cite{nielsen_chuang_2010} shown in \Cref{th:equivalence}.
    The formula is satisfiable if and only if there is an execution that violates the specification.
\end{proof}

%
The formula $\qmodel$ falls into the theory of nonlinear real arithmetic with trigonometric expressions, for which checking satisfiability is
undecidable \cite{Richardson68}. Yet, the $\delta$-relaxation of this problem is decidable \cite{DBLP:conf/cade/GaoAC12}.
That is why we use the $\delta$-satisfiability framework from \cite{DBLP:conf/cade/GaoKC13}, which is implemented in dReal\footnote{Available at \url{https://github.com/dreal/dreal4}}.
If the combined formula $\qmodel \land \neg \spec$ is found to be $\delta$-$\sat$, either it is indeed satisfiable (i.e., a counterexample has been found), or it is unsatisfiable (i.e., the program complies with the specification) but a $\delta$-perturbation on its numerical terms would satisfy the formula.
The parameter $\delta$ is user-controllable, and we show in the evaluation that the $\delta$-$\sat$ case for correct programs does not occur in practice for reasonable values of $\delta$.

While the $\delta$-relaxation must sacrifice completeness, it preserves soundness: If the formula is found to be unsatisfiable ($\unsat$), then the quantum program is indeed correct with respect to $\spec$.

\begin{theorem}[Soundness preservation]
    Let $\qmodel$ be the encoding of a quantum program model and $\spec$ be a specification. Assume that a $\delta$-satisfiability solver returns $\unsat$ for the formula $\qmodel \land \neg \spec$. Then the quantum program is correct.
\end{theorem}

\begin{proof}
    This follows from \Cref{sound-complete} and \cite{DBLP:conf/cade/GaoKC13}.
\end{proof}

\subsection{Running Example: Verification of Teleportation}
Coming up with the right specifications for quantum programs is not trivial.
Conveniently, as \tool maps all building blocks of quantum programs into an SMT representation, we have access to the full set of logic operators.

\medskip

We want our specification to express that teleportation has been successful, i.e., qubit $\ket{q_0}$ has moved to where qubit $\ket{q_2}$ was at the beginning (compare the right-hand side of \Cref{fig:tp}).
\begin{align*}
    \left(\ket{q_{5,2}} = \ket{q_{0,0}} \right) 
\end{align*}
This, however, is not the full specification. We need to disallow operations crossing the line between the first two qubits and the last one, which only becomes possible after measurement, where the classical communication channel can be used (cf.\ \Cref{subsec:tp}).
Therefore, we add an additional constraint that forbids state operations where these qubits appear together:
\begin{align*}
    \spec &= (\ket{q_{5,2}} = \ket{q_{0,0}}) \, \land \,\, \neg \exists 0 \leq i \leq 2\!:\ \to_i\!(\ket{q_{i,0}}\!, \ket{q_{i,2}}) 
    \ \lor \to_i\!(\ket{q_{i,1}}\!, \ket{q_{i,2}})
\end{align*}

Performing the verification is ``push-button,'' i.e., only requires writing the quantum program model and the specification. The corresponding Python code given in Appendix~\ref{app:programs} demonstrates that a user does not have to provide any proof steps as in previous works based on proof assistants.

\subsection{The \tool Over-Approximation}\label{sec:over-approximation}
\begin{wrapfigure}[9]{r}{0.43\textwidth}
    \vspace*{-15mm}
	\centering
	\scalebox{0.6}{
		\begin{tikzpicture}[scale=1.0]
		\def\r{2}
		\def\rr{2.4}
		
		\draw[thick] (\r/2,\r/3) node[circle,fill,inner sep=1,label=above:$\Large{\ket{q}}$] (psi) {};
		\draw (0,0) node[circle,fill,inner sep=1] (orig) {} -- (psi);
		\draw[dashed,thick] (orig) -- (\r/2,-\r/5) node (phi) {} -- (psi);
		
		\draw[thick] (orig) circle (\r);
		\draw[dashed,thick] (orig) ellipse (\r{} and \r/3);
		
		\draw[->,red,thick] (orig) -- ++(-\rr/5-0.1,-\rr/3-0.12) node[below right] (x) {$x$};
		\draw (-\r/5,-\r/3) (+) {};
		\draw (\r/5,\r/3) (-) {};
		
		\draw[->,blue,thick] (orig) -- ++(\rr+0.3,0) node[right] (y) {$y$};
		\draw (\r,0) (i) {};
		\draw (-\r,0)  (-i) {};
		
		\draw[->,forestgreen,thick] (orig) -- ++(0,\rr-0.1) node[above] (z) {$z$};
		\draw (0,\r)  (0) {};
		\draw (0,-\r) (1) {};
		
		\pic [draw=gray,text=gray,->,"$\phi$"] {angle = x--orig--phi};
		\pic [draw=gray,text=gray,<-,"$\theta$"] {angle = psi--orig--z};
		
		\draw[dashed] (-\r,0) -- ++(-\r/5,-\r/3) -- ++(\r,0) -- ++(\r,0);
		\draw[dashed] (-\r,0) -- ++(\r/5,\r/3);
		
		\draw[dashed] (\r,0) -- ++(-\r/5,-\r/3);
		\draw[dashed] (\r,0) -- ++(\r/5,\r/3) -- ++(-\r,0) -- ++(-\r,0);
		
		\draw[-] (-\r-\r/5,-\r/3) -- ++(0,\r) -- ++(\r,0) -- ++(\r,0) -- ++(\r/5,\r/3) -- ++(\r/5,\r/3) -- ++(-\r,0) -- ++(-\r,0) -- ++(-\r/5,-\r/3) -- ++(-\r/5,-\r/3);
		
		\draw[-] (-\r-\r/5,-\r/3) -- ++(0,-\r);
		\draw[-] (\r-\r/5,-\r/3) -- ++(0,-\r);

		\draw[-] (\r-\r/5,-\r/3) -- ++(0, -\r) -- ++(\r/5,\r/3) -- ++(\r/5,\r/3) -- ++(-\r,0) -- ++(-\r,0) -- ++(-\r/5,-\r/3) -- ++(-\r/5,-\r/3) -- ++(\r, 0) -- ++(\r, 0);
		
        \draw[-] (\r-\r/5,-\r/3) -- ++(0,\r);
        \draw[-] (\r+\r/5,\r/3) -- ++(0,\r);
        \draw[-] (\r+\r/5,\r/3) -- ++(0,-\r);

        \draw[-] (-\r+\r/5,\r/3) -- ++(0,-\r);
        \draw[-] (-\r+\r/5,\r/3) -- ++(0,\r);

		\end{tikzpicture}}
	\caption{The over-approximation visualized for a single qubit.}
	\label{fig:qubit-overapprox}
\end{wrapfigure}
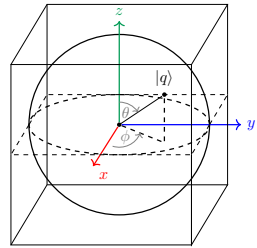

\Cref{enc:modeling-of-qubits} puts trigonometric functions into the SMT formula, which are computationally expensive. 
This can also be later seen in the evaluation.
Therefore, we introduce an over-approximation of the Hilbert space to make the verification task more efficient.
This is achieved via relaxing the qubit's degrees of freedom from the unit sphere to the unit box, visualized in \Cref{fig:qubit-overapprox}.

\begin{myenc}[Over-approximation]\label{enc:qubit-over-approx}
    We remove the constraints in \Cref{eq:qubit-periods} from \Cref{enc:modeling-of-qubits} and add the following constraint over the qubit's degrees of freedom:
    \begin{align}
        -1 \leq \alpha \leq 1 \ \spaceLand \ -1 \leq \beta_R \leq 1 \ \spaceLand \ -1 \leq \beta_I \leq 1.
    \end{align}
\end{myenc}

\section{Evaluation}

This section presents our experimental evaluation, demonstrating \tool's effectiveness in verifying several (correct) quantum programs that have symbolic inputs or symbolically parametrized quantum gates.

\subsection{Implementation}

\tool\footnote{Available for download at \url{doi.org/10.5281/zenodo.7400321}} is implemented as a Python library interfacing with dReal \cite{DBLP:conf/cade/GaoKC13} using about 5000 lines of code.
The \tool Python API allows users to specify the quantum program using a syntax inspired by Cirq \cite{cirq_developers_2021_5182845}.
The specification can be written using one of two formats: 
\begin{itemize}
    \item \textit{State vector:} One can specify assertions on any of the $2^n$ vector entries.
    \item \textit{Qubits:} One can specify assertions on any of the $n$ qubits.
\end{itemize}
The logic assertions use an SMT-LIB2-compatible Python API and support specifications expressing relationships between program inputs and outputs as well as intermediate states.

\subsection{Benchmark Problems and Setup}
An overview of the benchmark problems is given in \Cref{tab:benchmarks}. 
Further descriptions, including the specifications, are given in \Cref{app:programs}.

\begin{table}[t]
    \centering
    \begin{tabularx}{\linewidth}{@{}l p{5.7cm} llc@{}} \toprule
        \textbf{Program}& \textbf{Description}                  & \textbf{Depth\,} & \textbf{Input}         & \textbf{Parametrized} \\ \midrule
        Toffoli         & Toffoli Gate                          & 5             & Bit vector    & No \\
        TP              & Quantum Teleportation Circuit         & 6             & Infinite      & No  \\
        ADD-8           & 8-bit Quantum Adder                   & 48            & Bit vector    & No  \\
        QFT-$n$         & $n$-Qubit Quantum Fourier Transform   & $\mathcal{O}(n^2)$  & Bit vector    & No  \\
        QPE-$n$         & $n$-Bit Quantum Phase Estimation      & $\mathcal{O}(n^2)$  & Concrete      & Yes \\
        GDO-$n$         & $n$-Qubit Grover Diffusion Operator   & $\mathcal{O}(n)$  & Infinite      & No  \\
        \bottomrule
    \end{tabularx}
    \caption{Benchmark quantum programs for evaluating our verification procedure.
    ``Input'' describes the input space to the quantum programs and ``Parametrized'' expresses whether there
    are parametrized gates in the quantum program.
    }
    \label{tab:benchmarks}
    \vspace{-1em}
\end{table}

We compare our tool (``\tool'') against quantum simulation (``Simulation''), basic SMT solving based on linear algebra (``Basic SMT''), and \tool without over-approximation (``\tool (exact)'').

\begin{itemize}
    \item \textit{Simulation} is implemented in Qiskit \cite{Qiskit}. The technique enumerates all possible inputs to the quantum program and then compares the outputs with the specification. 
    We can only use this technique for a finite input space, i.e., for concrete and bit-vector inputs, but neither for symbolic qubits with the entire Hilbert space $\hilbert_2$ as input space, nor for parametrized gates.
    \item \textit{Basic SMT} is basic SMT solving using vectors and matrices, but not using direct mappings (\Cref{enc:mapping}).
    \item \tool (exact) is a modification of \tool where all over-approximation capabilities are removed, ending up with a technique that performs exact modeling, even when unnecessary (see \Cref{sec:over-approximation}).
\end{itemize}

We do not compare against the proof-assistant approaches \cite{DBLP:journals/corr/abs-1803-00699,DBLP:conf/birthday/LiuWZGLHDY17,DBLP:conf/cav/LiuZWYLLYZ19,DBLP:conf/esop/CharetonBBPV21,DBLP:conf/itp/Hietala0HL021} (cf. \Cref{sec:introduction}) because a comparison of run-times between an automated method, as implemented in \tool, and a semi-automated method relying on manual input is not meaningful. We also do not compare against \cite{DBLP:journals/corr/abs-1805-06908} because it neither supports the full gate set nor formal logic specifications.

The experiments use the value $\delta = 10^{-4}$.
We also compare the run-time of \tool for different precision levels $\delta$.

All experiments are carried out on a workstation with an AMD Ryzen ThreadRipper 3960X @ 3.8\,GHz $\times $ 24 cores processor and 256\,GB RAM.
The machine runs Ubuntu 20.04.3 LTS and each result is the average of 10 runs.

\begin{table*}[t]
    \centering
    \begin{tabular}{@{} p{2.9cm} @{\,\,\,}llll@{}}\toprule
    \textbf{Benchmark} & \textbf{Simulation} & \textbf{Basic SMT} & \textbf{\tool (exact)} & \textbf{\tool}  \\
    \midrule
    Toffoli   & \textbf{0.02 seconds}   & 11.1 seconds  & 1.3 seconds           & 0.4 seconds           \\
    TP        & N/A                     & 44.8 seconds  & \textbf{21.6 seconds} & 31.0 seconds          \\
    ADD-8     & 6.1 hours               & \memoryerror  & \textbf{7.6 seconds}  & 7.8 seconds           \\ \midrule
    QFT-3     & \textbf{0.005 seconds}  & 12.8 seconds  & 5.8 seconds           & 1.0 second            \\
    QFT-5     & \textbf{0.03 seconds}   & 17.6 minutes  & 2.6 minutes           & 26.4 seconds          \\
    QFT-10    & \textbf{1.5 seconds}    & 1.2 hours     & 10.9 hours            & 1.6 hours             \\
    QFT-12    & \textbf{14.0 seconds}   & 4.0 hours     & \timeout              & 7.4 hours             \\ \midrule
    QPE-3     & N/A                     & 19.2 seconds  & 34.0 seconds          & \textbf{8.7 seconds}  \\   
    QPE-5     & N/A                     & 18.2 minutes  & 42.3 minutes          & \textbf{3.9 minutes}  \\ \midrule
    GDO-5     & N/A                     & \timeout      & 9.2 seconds           & \textbf{1.3 seconds}  \\
    GDO-10    & N/A                     & \timeout      & 3.2 minutes           & \textbf{17.0 seconds} \\
    GDO-12    & N/A                     & \timeout      & 14.2 minutes          & \textbf{20.2 seconds} \\
    GDO-15    & N/A                     & \timeout      & 2.9 hours             & \textbf{1.0 minute}   \\
    GDO-18    & N/A                     & \timeout      & \timeout              & \textbf{4.9 minutes}  \\
    GDO-20    & N/A                     & \timeout      & \timeout              & \textbf{17.1 minutes} \\
    GDO-22    & N/A                     & \timeout      & \timeout              & \textbf{1.1 hours}    \\
    GDO-24    & N/A                     & \timeout      & \timeout              & \textbf{4.2 hours}    \\
    \bottomrule
    \end{tabular}
    \caption{Runtime comparison results for the benchmark problems described in \Cref{tab:benchmarks}. ``Simulation'' stands for simulation and enumeration of all cases. ``Basic SMT'' is SMT solving with full state and matrix construction. ``\tool (exact)'' is \tool where over-approximations have been removed. ``\tool'' (this work) utilizes a sound over-approximation. 
    ``N/A'' instances cannot be solved by simulation due to infinite state space.
    ``\memoryerror'' cases exceeded the available memory, and
    ``\timeout'' cases exceed the 12-hour time limit.
    }
    \label{tab:results}
\end{table*}

\begin{table}
    \centering
    \begin{tabular}{@{}l p{3cm} p{3cm} p{3cm} @{}}\toprule
        \textbf{Delta} & GDO-12 & GDO-15 & GDO-18\\
        \midrule
        $10^{-4}$ & 20.2 seconds    & 1.0 minutes   & 4.9 minutes\\
        $10^{-6}$ & 20.5 seconds    & 28.0 minutes  & 33.1 minutes  \\
        $10^{-8}$ & 20.8 seconds    & 49.4 minutes  & 58.7 minutes \\
        $10^{-10}$ & 21.1 seconds   & 52.3 minutes  & 1.2 hours \\
        \bottomrule
    \end{tabular}
    \caption{\tool run-time results for different precision values $\delta$.}
    \label{tab:delta-comp}
\end{table}

\subsection{Results}\label{subsec:results}

We summarize our results in \Cref{tab:results}.
\tool (exact) is best for quantum programs with concrete inputs or a small qubit count (TP and ADD-8); the over-approximation of \tool yields no speed-up for these instances.
\emph{Simulation} performs best for verifying combinatorial problems, i.e., for the quantum Fourier transform (QFT). Here, it can still feasibly enumerate a 12-qubit state space. Interestingly, Basic SMT scales best among the SMT-based procedures here; this is explained by the high amount of controlled operations, for which the mapping-based approach of \tool is inferior.

\smallskip

\tool offers a dramatic performance increase for quantum programs with symbolic inputs, i.e., quantum phase estimation (QPE) and Grover's diffusion operator (GDO). 
This highlights the advantage of over-approximation for this family of quantum programs.
Recall that simulation is not possible for both QPE and GDO, as that would require enumerating infinitely many inputs.

The precision value $\delta = 10^{-4}$ was sufficient for all benchmarks in our evaluation.
To investigate scalability in this parameter, \Cref{tab:delta-comp} compares the run-times for different values for GDO with 12, 15, and 18 qubits, respectively. For the higher qubit counts, the run-time increases significantly when we lower $\delta$ to $10^{-6}$, but then remains relatively stable when further tightening precision.

\medskip 

Overall, \tool is the strongest for quantum programs with infinite input space, i.e., programs where the (symbolic) input qubits can span the complete Hilbert space. 
Likewise, for programs that use parametrized quantum gates dependent on a symbolic parameter, \tool is the most effective.

\section{Discussion}

Symbolic execution and formal verification scale exponentially for the quantum case, as is the case for classical software.
That is to be expected: firstly, the simulation of quantum programs on classical hardware already takes exponential time and space due to the matrix representation of quantum mechanics, and secondly because the state space grows with every input variable added to the program.
Nonetheless, we have shown how to keep this exponential blow-up under control by introducing mappings and over-approximations.
In our evaluation, we symbolically executed quantum programs with up to 24 qubits. In comparison, even (concrete) quantum simulation for concrete inputs stops being feasible at around 30 qubits, requiring petabytes of main memory.
In conclusion, \tool is most effective for unknown inputs to the quantum programs or unknown parameters of quantum gates that therefore cannot be tested.

\section{Conclusion}
We introduced \tool, a symbolic verification technique that leverages over-approximation to make automated verification of quantum programs feasible.
We formalized quantum program semantics in SMT and proposed a sound over-approximation that allows scaling to realistic program sizes.
Thanks to the symbolic nature of our approach, we can analyze quantum programs with infinite input space, which is beyond the capabilities of quantum simulation.
We demonstrate these achievements by formally verifying multiple quantum programs against their specifications within a modest time frame.

In this paper, we focused on formalizing the mathematical foundations to model quantum programs, define specifications, and prove their specification compliance.
We intend this to be the first step in a larger, fully automated quantum verification framework, including counterexample-guided refinement. 
In the future, we will investigate strategies that allow us to verify hybrid programs that perform classical and quantum computations.

\bibliography{bibliography.bib}

\newpage
\appendix

\section{Appendix}

\subsection{Example: Gate mapping versus matrix encoding}\label{app:example}

The following is an additional example to illustrate the benefits and weaknesses of both the gate mapping and the matrix encoding of the CNOT gate:

\begin{example}[Controlled-\textsc{not} gate]\label{ex:controlled-not-gate}
    The controlled-\textsc{not} gate can be encoded as a \emph{direct mapping}
    \begin{align*}
        (\ket{q_0}\!, \ket{q_1}) \overset{U_{CX}}{\mapsto} \begin{cases} (\ket{q_0}\!, \ket{q_1}) & \text{if} \ket{q_0} = \ket{0}, \\ (\ket{q_0}\!, X(\ket{q_1})) & \text{if} \ket{q_0} = \ket{1}) \end{cases}
    \end{align*}

    This encoding only supports computational basis states for the control, as any other state would create an entangled state, which can only be encoded using a matrix representation:
    For a state $\ket{q_0 q_1} = \alpha_{00} \ket{00} + \alpha_{01} \ket{01} + \alpha_{10} \ket{10} + \alpha_{11} \ket{11}$, the controlled-\textsc{not} gate can be encoded as a matrix.
    The state vector of $\ket{q_0 q_1}$ is $(\alpha_{00}, \alpha_{01}, \alpha_{10}, \alpha_{11})^\top$. 
    Then,
    \begin{align*}
        U_{CX}(\ket{\psi}) &= U_{CX}\big( \alpha_{00} \ket{00} + \alpha_{01} \ket{01} + \alpha_{10} \ket{10} + \alpha_{11} \ket{11} \!\big)  \\
        &= \alpha_{00} \ket{00} +\alpha_{01} \ket{01} + \alpha_{11} \ket{10} + \alpha_{10} \ket{11}.
    \end{align*}
\end{example}

The following example shows that building multi-control gates is fairly straight-forward when using a mapping:

\begin{example}[Multi-control gate]\label{ex:multi-control-not-gate}
    Any gate $U$ can be equipped with multiple control qubits via a mapping
    \begin{align*}
        &U(\ket{q_n}).\mathtt{c}(\ket{q_0}, \ket{q_1}, \dots, \ket{q_{n-1}}) \mapsto \\
        &\begin{cases} 
        (\ket{q_0}, \ket{q_1}, \dots, \ket{q_n}) & \text{if } \bigvee_{i=0}^{n-1} \left( \ket{q_i} = \ket{0} \right), \\ 
        (\ket{q_0},\ket{q_1}, \dots, U(\ket{q_n})) & \text{if } \bigwedge_{i=0}^{n-1} \left( \ket{q_i} = \ket{1} \right),
        \end{cases}
    \end{align*}
    where function $\mathtt{c}(\cdot)$ binds the control qubits.
    Likewise, this encoding only supports computational basis states for control qubits \ket{q_0}, \ket{q_1}, \dots, \ket{q_{n-1}}.
\end{example}

\subsection{Proof of \Cref{th:equivalence}}\label{app:proof}

\begin{proof}
    \textbf{Base case:}
    Assume a quantum program model $\qmodel = (\mathcal{Q}, S, {\to,}\, \Theta, \mathcal{M}, V_0)$ over $n$ qubits, a state sequence of only two states $S = (s_0, s_1)$ and therefore only a single state operation ${\to} = (\to_1).$
    Further, without loss of generality, assume $\Theta = \mathcal{M} = \varnothing$ and $V_0 =\hilbert_2^n$, as the encoding only filters out executions that do not start in a state present in $\valuation$.

    To prove that this quantum program model is equivalent to the quantum circuit model,
    we begin with the qubit.
    The qubit encoding given in \Cref{enc:modeling-of-qubits} is equivalent to the standard qubit as it occurs in quantum circuits due to the restrictions on the phases, ensuring a magnitude of $|\alpha|^2 + |\beta|^2 = 1$.
    
    To prove equivalence of computations, there are three cases:
    \begin{enumerate}
        \item State operation $\to_1$ is a matrix: Assume an arbitrary input $\ket{\psi}$ to the quantum circuit. According to \Cref{enc:gate-matrix}, we can rename matrix $\to_1$ to $U$ and directly perform matrix multiplication to yield the circuit's output state $\ket{\psi'} = U \ket{\psi}$, which is equivalent to the final state $s_1$ of $\qmodel$, constructed the same way using \Cref{enc:matrix-multiplication}.
        \item State operation $\to_1$ is a mapping: \Cref{enc:mapping} defines a mapping between two sequences of qubits $Q = (q_0, \dots, q_{n-1}), Q' = (q'_0, \dots, q'_{n-1})$.
        Because the encoding restricts the mapping to bijections within Hilbert space, we maintain reversibility of the transformations by ensuring they are unitary, i.e., $\to_i^\dagger \cdot \to_i = I$, and we can express the mapping as a unitary $U$.
        Thus, the quantum circuit's output state is $\ket{\psi'} = U \bigkron_{i=0}^k q_i$.
        This directly extends to all mappings in \crefrange{enc:basic-single-qubit-gates}{enc:swap-gate}.
        \item State operation is a measurement: refer to the final paragraph.
    \end{enumerate}
    
    Now, allowing arbitrary $\Theta$, its elements parametrize state operations in the quantum program model and unitary gates in the circuit model, the only difference being that these parameters are symbolic in the former.
    We have proven equivalence of the computations in a single-operation quantum program model and the quantum circuit model.
    
    \smallskip
    
    \textbf{Induction step:}
    For the induction step, we form quantum program models with arbitrarily many operations: 
    We assume two quantum program models $\qmodel = (\mathcal{Q}, S, \to,\, \Theta, \varnothing, V_0)$ and $\qmodel' = (\mathcal{Q}, S', \to',\, \Theta, \varnothing$, $V_0)$,
    with unrestricted initial valuations $V_0 = \hilbert_2^n$, because restricting $V_0$, without loss of generality, only filters out computations that do no start in a specified valuation.
    Both models operate on the same set of qubits $\mathcal{Q}$.
    
    $\qmodel$ contains arbitrarily many operations and $\qmodel'$ contains one
    operation:
    $S = (s_0, \dots, s_{m-1}),$ $\to = (\to_1, \dots, \to_{m-1})$ and $S' = (s_{m-1}, s_m)$, $\to' = (\to_{m})$.
    We immediately observe that this 
    allows us to merge both models where the output of the first is input to the second, forming
    the combined quantum program model $\qmodel^* = (\mathcal{Q}, S^*, \to^*, \Theta, \varnothing, V_0)$
    where $S^* = (s_0, \dots, s_{m})$ and $\to^* = (\to_1, \dots, \to_m)$.
    
    \smallskip
    
    Finally, we prove the equivalence of measurements.
    We use the fact that for any quantum circuit, measurement can be deferred to the end of the computation \cite{nielsen_chuang_2010}.
    In the quantum program model, measurement can occur in any state of the state sequence, creating branching states for all possible measurement outcomes by applying the general measurement matrices.
    Now, without loss of generality, assume that quantum program model $\qmodel$ measures each qubit at most once.
    This allows us to also defer measurements to the last state of the quantum program model.
    Using the equivalence between matrix state operations in $\qmodel$ and matrices in the circuit model and the (non-unitary) matrix representation of quantum measurements \cite{nielsen_chuang_2010}, equivalence between measurements in both representations becomes obvious.
    
    This structural extension allows for building arbitrarily deep quantum program models, concluding the induction step.

    We therefore have proven that we can map between quantum program models and circuits with arbitrarily many gates.
\end{proof}

\subsection{Further Benchmark Information}\label{app:programs}

We describe the benchmark problems in more detail, giving the quantum program model, specification, and \tool code.

\subsubsection{Teleportation}

Quantum teleportation (TP) is used as running example throughout the main body of this paper, cf. \Cref{subsec:tp}.
Performing the verification is ''push-button´´, i.e., only requires writing the quantum program and specification. The corresponding Python code is given in \Cref{lst:tp}.

\pythonexternal[caption={Quantum teleportation verification code in \tool.},captionpos=b,label=lst:tp]{code/tp.py}

\subsubsection{Grover's Diffusion Operator}
GDO is the amplitude amplification component of Grover's algorithm \cite{DBLP:conf/stoc/Grover96}, offering quadratic speedup over classical unordered database search.
The circuit diagram is shown in \Cref{fig:gdo} for 3 qubits.
GDO is defined as $U_\psi = 2 \ket{\psi}\bra{\psi} - I_n$, where $\ket{\psi}$ is the uniform superposition over $2^n$ basis states $\ket{\psi} = \frac{1}{2^n}\sum_{i=0}^{2^n} \ket{i}$ and $I_n$ is the $2^n \times 2^n$ identity matrix.

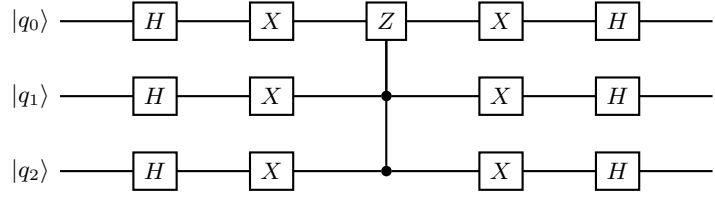
\begin{figure}[H]
    \centering
    \scalebox{1.0}{
    \begin{quantikz}[column sep=3.0em]
    \lstick{\ket{q_0}}
    & \gate{H}  
    & \gate{X}
    & \gate{Z}   
    & \gate{X} 
    & \gate{H}
    & \qw \\
    \lstick{\ket{q_1}}
    & \gate{H}  
    & \gate{X}
    & \ctrl{-1}   
    & \gate{X} 
    & \gate{H}
    & \qw \\
    \lstick{\ket{q_2}}                              
    & \gate{H}  
    & \gate{X}
    & \ctrl{-2}      
    & \gate{X} 
    & \gate{H}
    & \qw \\
    \end{quantikz}
    }
    \vspace{-1.5em}
    \caption{Grover's diffusion operator, adapted from \cite{DBLP:conf/stoc/Grover96}.}
    \label{fig:gdo}
\end{figure}

For Grover's diffusion operator (GDO), we define its quantum program model as
$\qmodel = (\mathcal{Q}, S, \to, \Theta, \mathcal{M}, \valuation)$, where
\begin{align*}
    \mathcal{Q} =&\ \{ \ket{q_0}, \dots, \ket{q_{n-1}} \}, \\
    S =&\ (s_0, \dots, s_F),  \\
    {\to} =&\ \big(H(\mathcal{Q}), X(\mathcal{Q}), \\
    &\ \phantom{(} Z(\ket{q_0}).\mathtt{c}(\ket{q_1}, \dots, \ket{q_{n-1}}), X(\mathcal{Q}), H(\mathcal{Q}) \big), \\ 
    \Theta =&\ \varnothing, \\
    \mathcal{M} =&\ \varnothing, \\
    \valuation =&\ \hilbert_2^n.
\end{align*}

The specification for GDO describes successful amplitude amplification for input qubits with negative phase:
\begin{align*}
    \spec &= \left( s_{1,0} \leq 0 \,?\, s_{1,F} \leq s_{1,0} : \textsc{True} \right) \, \land \\
    & \dots \land \left( s_{n,0} \leq 0 \,?\, s_{n,F} \leq s_{n,0} : \textsc{True} \right)
\end{align*}

The code is shown in \Cref{lst:gdo}.

\begin{lstfloat}[H]
\pythonexternalsmall[caption={Verification code for Grover's diffusion operator in \tool.},captionpos=b,label=lst:gdo]{code/gdo.py}
\end{lstfloat}

\subsubsection{Quantum Fourier Transform}

The quantum Fourier transform is a basic building block of many important quantum algorithms, for example quantum phase estimation (QPE), which is described in the next paragraph. The circuit diagram for the QFT is shown in \Cref{fig:qft}.
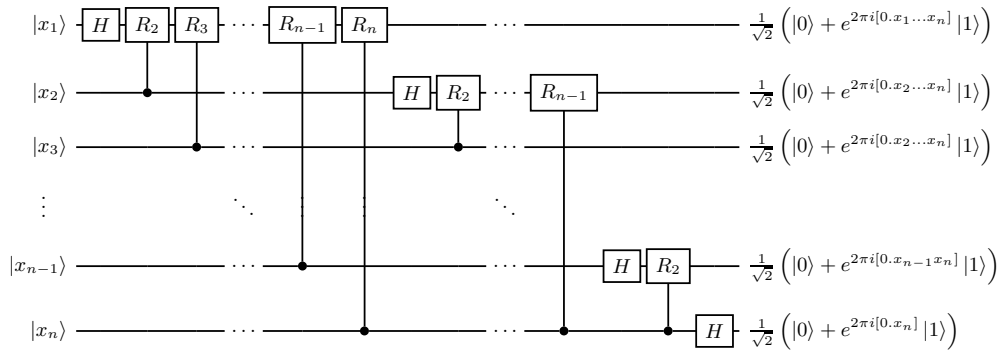
\begin{figure*}
    \centerline{
    \scalebox{0.85}{
    \begin{quantikz}[column sep=1mm]
    \lstick{$\ket{x_1}$}
    &\gate{H}
    & \gate{R_2}
    & \gate{R_3} 
    & \ \ldots\ \qw
    & \gate{R_{n-1}} 
    & \gate{R_n} 
    & \qw 
    & \qw
    & \ \ldots\ \qw
    & \qw
    & \qw
    & \qw
    & \qw
    & \rstick{$\frac{1}{\sqrt{2}} \left(\ket{0} + e^{2 \pi i [0.x_1 \dots x_n]}\ket{1}\right)$} \qw \\
    \lstick{$\ket{x_2}$} \qw 
    & \qw
    & \ctrl{-1}
    & \qw
    & \ \ldots\ \qw
    & \qw 
    & \qw
    & \gate{H}
    & \gate{R_2} 
    & \ \ldots\ \qw
    & \gate{R_{n-1}}
    & \qw 
    & \qw 
    & \qw 
    & \rstick{$\frac{1}{\sqrt{2}} \left(\ket{0} + e^{2 \pi i [0.x_2 \dots x_n]}\ket{1}\right)$} \qw \\
    \lstick{$\ket{x_3}$} \qw
    & \qw 
    & \qw
    & \ctrl{-2}
    & \ \ldots\ \qw
    & \qw 
    & \qw
    & \qw 
    & \ctrl{-1} 
    & \ \ldots\ \qw
    & \qw 
    & \qw
    & \qw 
    & \qw
    & \rstick{$\frac{1}{\sqrt{2}} \left(\ket{0} + e^{2 \pi i [0.x_2 \dots x_n]}\ket{1}\right)$} \qw \\
    \lstick{\vdots \quad} &&&& \ddots & \vdots & \vdots & & & \ddots & & \\
    \lstick{$\ket{x_{n-1}}$} \qw
    & \qw
    & \qw 
    & \qw
    & \ \ldots\ \qw
    & \ctrl{-4}
    & \qw 
    & \qw 
    & \qw
    & \ \ldots\ \qw
    & \qw
    & \gate{H} 
    & \gate{R_2} 
    & \qw
    & \rstick{$\frac{1}{\sqrt{2}} \left(\ket{0} + e^{2 \pi i [0.x_{n-1}x_n]}\ket{1}\right)$} \qw \\
    \lstick{$\ket{x_n}$} \qw
    & \qw 
    & \qw
    & \qw
    & \ \ldots\ \qw
    & \qw
    & \ctrl{-5} 
    & \qw
    & \qw
    & \ \ldots\ \qw
    & \ctrl{-4}
    & \qw 
    & \ctrl{-1}
    & \gate{H}
    & \rstick{$\frac{1}{\sqrt{2}} \left(\ket{0} + e^{2 \pi i [0.x_n]}\ket{1}\right)$} \qw \\
    \end{quantikz}
    }
    }
    \caption{The quantum Fourier transform, adapted from \cite{nielsen_chuang_2010}}
    \label{fig:qft}
\end{figure*}

We introduce the direct mapping for the parametrized phase shift gate.

\begin{myenc2}[Parametrized phase shift gate]
    The phase shift gate $R_k$, which depends on a symbolic parameter $k$, performs a rotation around the $z$ axis. Its mapping is defined for the amplitudes:
    $$
        R_k \left( \mat{\alpha \\ \beta} \right) \defeq \mat{\alpha \\ e^{\frac{2 \pi i}{2^k}} \beta}.
    $$
\end{myenc2}

The model for the QFT is
$\qmodel = (\mathcal{Q}, S, \to, \varnothing, \varnothing, \valuation)$, where
\begin{align*}
    \mathcal{Q} &= \{ \ket{x_1}, \dots, \ket{x_n} \}, \\
    S &= (s_0, \dots, s_F),  \\
    \to &= \big( H(\ket{x_1}), R_2(\ket{x_2}, \ket{x_1}),\\
    &\,\, R_3(\ket{x_3}\!, \ket{x_1}) \dots R_{n-1}(\ket{x_{n-1}}\!, \ket{x_1}), R_{n}(\ket{x_n}\!, \ket{x_1}), \\
    &\quad \ \ H(\ket{x_2}), R_2(\ket{x_3}\!, \ket{x_2}), \dots, R_{n-1}(\ket{x_n}\!, \ket{x_2})  \\
    &\quad \ \ H(\ket{x_{n-1}}), R_2(\ket{x_n}\!, \ket{x_{n-1}}) \\
    &\quad \ \ H(\ket{x_n}) \big), \\ 
    \valuation &= ( \{\ket{0}\!, \ket{1}\}, \dots, \{\ket{0}\!, \ket{1}\} ),
\end{align*}
with $F = (n^2 + n)/2$, and the specification is
\begin{align*}
    \spec = \left( \bigwedge_{i=1}^n \left( \ket{x_{(i,F)}} = \frac{1}{\sqrt{2}} \left(\ket{0} + e^{2 \pi i [0.x_i \dots x_n]} \right) \right) \right).
\end{align*}

The model $\qmodel$ and the negated specification $\neg\spec$ are used to analyze the correctness of the quantum program, with code given in \Cref{lst:qft}.

\pythonexternalsmall[caption={Quantum Fourier transform verification code in \tool.},captionpos=b,label=lst:qft]{code/qft.py}

\subsubsection{Quantum Phase Estimation}

Quantum phase estimation is, unlike the other benchmarks, a program on a concrete input because it 
takes the zero vector as input.
However, it depends on a parameter $\theta$, a real number where $0 \leq \theta \pi$.
The model is $\qmodel = (\mathcal{Q}, S, \to, \Theta, \varnothing, \valuation)$, where
\begin{align*}
    \mathcal{Q} &= \{ \ket{q_0}, \dots, \ket{q_{n-1}}, \ket{u} \}, \\
    S &= (s_0, \dots, s_F),  \\
    \to &= \big(H(\ket{q_0}\!, \dots, \ket{q_{n-1}}), 
    P\big(\ket{0}\!, (2 \theta \pi) ^{2^{n-1}}\big).\mathtt{c}(\ket{q_0}), \\
    & P\big(\ket{0}\!, (2 \theta \pi) ^{2^{n-2}}\big).\mathtt{c}(\ket{q_1}), 
    \dots, 
    P\big(\ket{0}\!, (2 \theta \pi) ^{2^0}\big).\mathtt{c}(\ket{q_{n-1}}), \\
    & \mathit{SWAP}(\ket{q_0}\!, \ket{q_{n-1}}), \mathit{SWAP}(\ket{q_1}\!, \ket{q_{n-2}}), \dots, \\ & \mathit{SWAP}(\big|q_{\lfloor \frac{n}{2} \rfloor - 1}\big\rangle, \big|q_{n-{\lfloor \frac{n}{2} \rfloor}} \big\rangle), \\
    & QFT^\dagger(\ket{q_0}, \dots, \ket{q_{n-1}})
    \big), \\ 
    \Theta &= \{ \theta \}, \\
    \valuation &= \{0\}^{n+1}
\end{align*}

and $QFT^\dagger$ is the inverse of the quantum Fourier transform model.
Let $2^n \theta = a + 2^n \delta$ where $a$ is the nearest integer to $2^n \theta$.
Then, the specification can be given as
\begin{align*}
    \spec = \bigkron_{i = 0}^{n-1} \ket{q_{(i, F)}} = [\alpha_1 \alpha_2 \dots \alpha_{2^n}] \land \alpha_a \geq \frac{4}{\pi^2}.
\end{align*}

The code given in \Cref{lst:qpe}.

\pythonexternalsmall[caption={Quantum phase estimation verification code in \tool.},captionpos=b,label=lst:qpe]{code/qpe.py}

\subsection{Complete SMT Encoding of Running Example}\label{sec:complete_smt_encoding_teleportation}

The following tables give the complete SMT encoding that \tool automatically generates from the code (\Cref{lst:gdo}) provided by the developer of the quantum program.

\begin{enumerate}
    \item \Cref{tab:tp-smt-verbose-states} lists all the symbolic variables per state defined in the quantum program model $\qmodel$ (\Cref{def:quantum-program-model}).
    For each state, \tool generates the corresponding qubit symbols via \Cref{enc:modeling-of-qubits} and \Cref{enc:quantum-measurement}.
    \item \Cref{tab:tp-smt-verbose-qubit-constraints} lists the qubit constraints via \Cref{eq:qubit-degrees-of-freedom} and \Cref{eq:qubit-periods} of \Cref{enc:modeling-of-qubits} that are imposed on the qubits in \Cref{tab:tp-smt-verbose-states}.
    \item \Cref{tab:tp-smt-verbose-initial} lists the assertions specifying the initial valuations $V_0$ defined in the quantum program model $\qmodel$.
    We can see that if the initial valuation of a qubit is the complete Hilbert space $\hilbert_2$, no assertion is put in place.
    \item \Cref{tab:tp-smt-verbose-operations} lists the assertions that specify the program operations and measurements $\to$ defined in the quantum program model $\qmodel$. 
    We can see that, for some steps, direct mappings (\Cref{enc:mapping}) are applied, while for others, especially the entangling steps, a gate matrix (\Cref{enc:gate-matrix}) is used which requires building a state vector (here, $s_1$).
\end{enumerate}

\Cref{tab:tp-smt-verbose-specification} encodes the specification $\varphi$.
The constraint that forbids state operations to cross the line between the first two qubits and the last one before measurement has taken place (cf.\ \Cref{subsec:tp}) is implemented in a preprocessing step before expanding to the SMT formula.

\tiny

\begin{tabularx}{1\textwidth}{@{}lX@{}}\toprule
    \textbf{State}  &       \textbf{Variables} \\ \midrule
    $s_0$           &       $\alpha_{0,0}, \beta_{R_{0,0}}, \beta_{I_{0,0}}, \phi_{0,0}, \theta_{0,0}$,
                            $\alpha_{0,1}, \beta_{R_{0,1}}, \beta_{I_{0,1}}, \phi_{0,1}, \theta_{0,1}$, 
                            $\alpha_{0,2}, \beta_{R_{0,2}}, \beta_{I_{0,2}}, \phi_{0,2}, \theta_{0,2}$ \\
    $s_1$           &       $\alpha_{1,0}, \beta_{R_{1,0}}, \beta_{I_{1,0}}, \phi_{1,0}, \theta_{1,0}$,
                            $\alpha_{1,1}, \beta_{R_{1,1}}, \beta_{I_{1,1}}, \phi_{1,1}, \theta_{1,1}$, 
                            $\alpha_{1,2}, \beta_{R_{1,2}}, \beta_{I_{1,2}}, \phi_{1,2}, \theta_{1,2}$ \\
    $s_2$           &       $\alpha_{2,0}, \beta_{R_{2,0}}, \beta_{I_{2,0}}, \phi_{2,0}, \theta_{2,0}$,
                            $\alpha_{2,1}, \beta_{R_{2,1}}, \beta_{I_{2,1}}, \phi_{2,1}, \theta_{2,1}$, 
                            $\alpha_{2,2}, \beta_{R_{2,2}}, \beta_{I_{2,2}}, \phi_{2,2}, \theta_{2,2}$ \\
    $s_3(00)$\,\,   &       $\alpha_{3,0}(00), \beta_{R_{3,0}}(00), \beta_{I_{3,0}}(00), \phi_{3,0}(00), \theta_{3,0}(00)$,
                            $\alpha_{3,1}(00), \beta_{R_{3,1}}(00), \beta_{I_{3,1}}(00), \phi_{3,1}(00), \theta_{3,1}(00)$, 
                            $\alpha_{3,2}(00), \beta_{R_{3,2}}(00), \beta_{I_{3,2}}(00), \phi_{3,2}(00), \theta_{3,2}(00)$ \\     
    $s_3(01)$       &       $\alpha_{3,0}(01), \beta_{R_{3,0}}(01), \beta_{I_{3,0}}(01), \phi_{3,0}(01), \theta_{3,0}(01)$,
                            $\alpha_{3,1}(01), \beta_{R_{3,1}}(01), \beta_{I_{3,1}}(01), \phi_{3,1}(01), \theta_{3,1}(01)$, 
                            $\alpha_{3,2}(01), \beta_{R_{3,2}}(01), \beta_{I_{3,2}}(01), \phi_{3,2}(01), \theta_{3,2}(01)$ \\ 
    $s_3(10)$       &       $\alpha_{3,0}(10), \beta_{R_{3,0}}(10), \beta_{I_{3,0}}(10), \phi_{3,0}(10), \theta_{3,0}(10)$,
                            $\alpha_{3,1}(10), \beta_{R_{3,1}}(10), \beta_{I_{3,1}}(10), \phi_{3,1}(10), \theta_{3,1}(10)$, 
                            $\alpha_{3,2}(10), \beta_{R_{3,2}}(10), \beta_{I_{3,2}}(10), \phi_{3,2}(10), \theta_{3,2}(10)$ \\ 
    $s_3(11)$       &       $\alpha_{3,0}(11), \beta_{R_{3,0}}(11), \beta_{I_{3,0}}(11), \phi_{3,0}(11), \theta_{3,0}(11)$,
                            $\alpha_{3,1}(11), \beta_{R_{3,1}}(11), \beta_{I_{3,1}}(11), \phi_{3,1}(11), \theta_{3,1}(11)$, 
                            $\alpha_{3,2}(11), \beta_{R_{3,2}}(11), \beta_{I_{3,2}}(11), \phi_{3,2}(11), \theta_{3,2}(11)$ \\ 
    $s_4(00)$       &       $\alpha_{4,0}(00), \beta_{R_{4,0}}(00), \beta_{I_{4,0}}(00), \phi_{4,0}(00), \theta_{4,0}(00)$,
                            $\alpha_{4,1}(00), \beta_{R_{4,1}}(00), \beta_{I_{4,1}}(00), \phi_{4,1}(00), \theta_{4,1}(00)$, 
                            $\alpha_{4,2}(00), \beta_{R_{4,2}}(00), \beta_{I_{4,2}}(00), \phi_{4,2}(00), \theta_{4,2}(00)$ \\     
    $s_4(01)$       &       $\alpha_{4,0}(01), \beta_{R_{4,0}}(01), \beta_{I_{4,0}}(01), \phi_{4,0}(01), \theta_{4,0}(01)$,
                            $\alpha_{4,1}(01), \beta_{R_{4,1}}(01), \beta_{I_{4,1}}(01), \phi_{4,1}(01), \theta_{4,1}(01)$, 
                            $\alpha_{4,2}(01), \beta_{R_{4,2}}(01), \beta_{I_{4,2}}(01), \phi_{4,2}(01), \theta_{4,2}(01)$ \\ 
    $s_4(10)$       &       $\alpha_{4,0}(10), \beta_{R_{4,0}}(10), \beta_{I_{4,0}}(10), \phi_{4,0}(10), \theta_{4,0}(10)$,
                            $\alpha_{4,1}(10), \beta_{R_{4,1}}(10), \beta_{I_{4,1}}(10), \phi_{4,1}(10), \theta_{4,1}(10)$, 
                            $\alpha_{4,2}(10), \beta_{R_{4,2}}(10), \beta_{I_{4,2}}(10), \phi_{4,2}(10), \theta_{4,2}(10)$ \\ 
    $s_4(11)$       &       $\alpha_{4,0}(11), \beta_{R_{4,0}}(11), \beta_{I_{4,0}}(11), \phi_{4,0}(11), \theta_{4,0}(11)$,
                            $\alpha_{4,1}(11), \beta_{R_{4,1}}(11), \beta_{I_{4,1}}(11), \phi_{4,1}(11), \theta_{4,1}(11)$, 
                            $\alpha_{4,2}(11), \beta_{R_{4,2}}(11), \beta_{I_{4,2}}(11), \phi_{4,2}(11), \theta_{4,2}(11)$ \\ 
    $s_5(00)$       &       $\alpha_{5,0}(00), \beta_{R_{5,0}}(00), \beta_{I_{5,0}}(00), \phi_{5,0}(00), \theta_{5,0}(00)$,
                            $\alpha_{5,1}(00), \beta_{R_{5,1}}(00), \beta_{I_{5,1}}(00), \phi_{5,1}(00), \theta_{5,1}(00)$, 
                            $\alpha_{5,2}(00), \beta_{R_{5,2}}(00), \beta_{I_{5,2}}(00), \phi_{5,2}(00), \theta_{5,2}(00)$ \\     
    $s_5(01)$       &       $\alpha_{5,0}(01), \beta_{R_{5,0}}(01), \beta_{I_{5,0}}(01), \phi_{5,0}(01), \theta_{5,0}(01)$,
                            $\alpha_{5,1}(01), \beta_{R_{5,1}}(01), \beta_{I_{5,1}}(01), \phi_{5,1}(01), \theta_{5,1}(01)$, 
                            $\alpha_{5,2}(01), \beta_{R_{5,2}}(01), \beta_{I_{5,2}}(01), \phi_{5,2}(01), \theta_{5,2}(01)$ \\ 
    $s_5(10)$       &       $\alpha_{5,0}(10), \beta_{R_{5,0}}(10), \beta_{I_{5,0}}(10), \phi_{5,0}(10), \theta_{5,0}(10)$,
                            $\alpha_{5,1}(10), \beta_{R_{5,1}}(10), \beta_{I_{5,1}}(10), \phi_{5,1}(10), \theta_{5,1}(10)$, 
                            $\alpha_{5,2}(10), \beta_{R_{5,2}}(10), \beta_{I_{5,2}}(10), \phi_{5,2}(10), \theta_{5,2}(10)$ \\ 
    $s_5(11)$       &       $\alpha_{5,0}(11), \beta_{R_{5,0}}(11), \beta_{I_{5,0}}(11), \phi_{5,0}(11), \theta_{5,0}(11)$,
                            $\alpha_{5,1}(11), \beta_{R_{5,1}}(11), \beta_{I_{5,1}}(11), \phi_{5,1}(11), \theta_{5,1}(11)$, 
                            $\alpha_{5,2}(11), \beta_{R_{5,2}}(11), \beta_{I_{5,2}}(11), \phi_{5,2}(11), \theta_{5,2}(11)$ \\ 
    \bottomrule
    \caption{Full SMT formula of teleportation: states.}
    \label{tab:tp-smt-verbose-states}
\end{tabularx}

\begin{tabularx}{1.05\linewidth}{@{}ccX@{}}\toprule
    \textbf{State} & \textbf{Qubit} & \textbf{Assertions} \\ \midrule
    \rot{$s_0$}
    & \rot{$\ket{q_{0,0}}$}
    & $\alpha_{0,0} = \cos\frac{\theta_{0,0}}{2}, \beta_{R_{0,0}} = \cos\phi_{0,0} \sin\frac{\theta_{0,0}}{2}, \beta_{I_{0,0}} = \sin\phi_{0,0} \sin\frac{\theta_{0,0}}{2},$ \\
    &               & $0 \leq \theta_{0,0} \leq \pi, 0 \leq \phi_{0,0} < 2\pi, \theta_{0,0} = 0 \Rightarrow \phi_{0,0} = 0, \theta_{0,0} = \pi \Rightarrow \phi_{0,0} = 0$. \\
    & \rot{$\ket{q_{0,1}}$}
    & $\alpha_{0,1} = \cos\frac{\theta_{0,1}}{2}, \beta_{R_{0,1}} = \cos\phi_{0,1} \sin\frac{\theta_{0,1}}{2}, \beta_{I_{0,1}} = \sin\phi_{0,1} \sin\frac{\theta_{0,1}}{2},$ \\
    &                & $0 \leq \theta_{0,1} \leq \pi, 0 \leq \phi_{0,1} < 2\pi, \theta_{0,1} = 0 \Rightarrow \phi_{0,1} = 0, \theta_{0,1} = \pi \Rightarrow \phi_{0,1} = 0$. \\
    & \rot{$\ket{q_{0,2}}$} 
    & $\alpha_{0,2} = \cos\frac{\theta_{0,2}}{2}, \beta_{R_{0,2}} = \cos\phi_{0,2} \sin\frac{\theta_{0,2}}{2}, \beta_{I_{0,2}} = \sin\phi_{0,2} \sin\frac{\theta_{0,2}}{2},$ \\
    &                & $0 \leq \theta_{0,2} \leq \pi, 0 \leq \phi_{0,2} < 2\pi, \theta_{0,2} = 0 \Rightarrow \phi_{0,2} = 0, \theta_{0,2} = \pi \Rightarrow \phi_{0,2} = 0$. \\ \midrule
    \rot{$s_1$} 
    & \rot{$\ket{q_{1,0}}$} & $\alpha_{1,0} = \cos\frac{\theta_{1,0}}{2}, \beta_{R_{1,0}} = \cos\phi_{1,0} \sin\frac{\theta_{1,0}}{2}, \beta_{I_{1,0}} = \sin\phi_{1,0} \sin\frac{\theta_{1,0}}{2},$ \\
    &               & $0 \leq \theta_{1,0} \leq \pi, 0 \leq \phi_{1,0} < 2\pi, \theta_{1,0} = 0 \Rightarrow \phi_{1,0} = 0, \theta_{1,0} = \pi \Rightarrow \phi_{1,0} = 0$. \\
    & \rot{$\ket{q_{1,1}}$} & $\alpha_{1,1} = \cos\frac{\theta_{1,1}}{2}, \beta_{R_{1,1}} = \cos\phi_{1,1} \sin\frac{\theta_{1,1}}{2}, \beta_{I_{1,1}} = \sin\phi_{1,1} \sin\frac{\theta_{1,1}}{2},$ \\
    &                & $0 \leq \theta_{1,1} \leq \pi, 0 \leq \phi_{1,1} < 2\pi, \theta_{1,1} = 0 \Rightarrow \phi_{1,1} = 0, \theta_{1,1} = \pi \Rightarrow \phi_{1,1} = 0$. \\
    & \rot{$\ket{q_{1,2}}$} & $\alpha_{1,2} = \cos\frac{\theta_{1,2}}{2}, \beta_{R_{1,2}} = \cos\phi_{1,2} \sin\frac{\theta_{1,2}}{2}, \beta_{I_{1,2}} = \sin\phi_{1,2} \sin\frac{\theta_{1,2}}{2},$ \\
    &                & $0 \leq \theta_{1,2} \leq \pi, 0 \leq \phi_{1,2} < 2\pi, \theta_{1,2} = 0 \Rightarrow \phi_{1,2} = 0, \theta_{1,2} = \pi \Rightarrow \phi_{1,2} = 0$. \\ \midrule
    \rot{$s_2$}
    & \rot{$\ket{q_{2,0}}$} & $\alpha_{2,0} = \cos\frac{\theta_{2,0}}{2}, \beta_{R_{2,0}} = \cos\phi_{2,0} \sin\frac{\theta_{2,0}}{2}, \beta_{I_{2,0}} = \sin\phi_{2,0} \sin\frac{\theta_{2,0}}{2},$ \\
    &               & $0 \leq \theta_{2,0} \leq \pi, 0 \leq \phi_{2,0} < 2\pi, \theta_{2,0} = 0 \Rightarrow \phi_{2,0} = 0, \theta_{2,0} = \pi \Rightarrow \phi_{2,0} = 0$. \\
    & \rot{$\ket{q_{2,1}}$} & $\alpha_{2,1} = \cos\frac{\theta_{2,1}}{2}, \beta_{R_{2,1}} = \cos\phi_{2,1} \sin\frac{\theta_{2,1}}{2}, \beta_{I_{2,1}} = \sin\phi_{2,1} \sin\frac{\theta_{2,1}}{2},$ \\
    &                & $0 \leq \theta_{2,1} \leq \pi, 0 \leq \phi_{2,1} < 2\pi, \theta_{2,1} = 0 \Rightarrow \phi_{2,1} = 0, \theta_{2,1} = \pi \Rightarrow \phi_{2,1} = 0$. \\
    & \rot{$\ket{q_{2,2}}$} & $\alpha_{2,2} = \cos\frac{\theta_{2,2}}{2}, \beta_{R_{2,2}} = \cos\phi_{2,2} \sin\frac{\theta_{2,2}}{2}, \beta_{I_{2,2}} = \sin\phi_{2,2} \sin\frac{\theta_{2,2}}{2},$ \\
    &                & $0 \leq \theta_{2,2} \leq \pi, 0 \leq \phi_{2,2} < 2\pi, \theta_{2,2} = 0 \Rightarrow \phi_{2,2} = 0, \theta_{2,2} = \pi \Rightarrow \phi_{2,2} = 0$. \\ \midrule
    \rot{$s_3(00)$} 
    & \rot{$\ket{q_{3,0}(00)}$} & $\alpha_{3,0}(00) = \cos\frac{\theta_{3,0}(00)}{2}, \beta_{R_{3,0}}(00) = \cos\phi_{3,0}(00) \sin\frac{\theta_{3,0}(00)}{2}, \beta_{I_{3,0}}(00) = \sin\phi_{3,0}(00) \sin\frac{\theta_{3,0}(00)}{2},$ \\
    &               & $0 \leq \theta_{3,0}(00) \leq \pi, 0 \leq \phi_{3,0}(00) < 2\pi, \theta_{3,0}(00) = 0 \Rightarrow \phi_{3,0}(00) = 0, \theta_{3,0}(00) = \pi \Rightarrow \phi_{3,0}(00) = 0$. \\
    & \rot{$\ket{q_{3,1}(00)}$} & $\alpha_{3,1}(00) = \cos\frac{\theta_{3,1}(00)}{2}, \beta_{R_{3,1}}(00) = \cos\phi_{3,1}(00) \sin\frac{\theta_{3,1}(00)}{2}, \beta_{I_{3,1}}(00) = \sin\phi_{3,1}(00) \sin\frac{\theta_{3,1}(00)}{2},$ \\
    &                & $0 \leq \theta_{3,1}(00) \leq \pi, 0 \leq \phi_{3,1}(00) < 2\pi, \theta_{3,1}(00) = 0 \Rightarrow \phi_{3,1}(00) = 0, \theta_{3,1}(00) = \pi \Rightarrow \phi_{3,1}(00) = 0$. \\
    & \rot{$\ket{q_{3,2}(00)}$} & $\alpha_{3,2}(00) = \cos\frac{\theta_{3,2}(00)}{2}, \beta_{R_{3,2}}(00) = \cos\phi_{3,2}(00) \sin\frac{\theta_{3,2}(00)}{2}, \beta_{I_{3,2}}(00) = \sin\phi_{3,2}(00) \sin\frac{\theta_{3,2}(00)}{2},$ \\
    &                & $0 \leq \theta_{3,2}(00) \leq \pi, 0 \leq \phi_{3,2}(00) < 2\pi, \theta_{3,2}(00) = 0 \Rightarrow \phi_{3,2}(00) = 0, \theta_{3,2}(00) = \pi \Rightarrow \phi_{3,2}(00) = 0$. \\ \midrule
    \rot{$s_3(01)$} 
    & \rot{$\ket{q_{3,0}(01)}$} & $\alpha_{3,0}(01) = \cos\frac{\theta_{3,0}(01)}{2}, \beta_{R_{3,0}}(01) = \cos\phi_{3,0}(01) \sin\frac{\theta_{3,0}(01)}{2}, \beta_{I_{3,0}}(01) = \sin\phi_{3,0}(01) \sin\frac{\theta_{3,0}(01)}{2},$ \\
    &               & $0 \leq \theta_{3,0}(01) \leq \pi, 0 \leq \phi_{3,0}(01) < 2\pi, \theta_{3,0}(01) = 0 \Rightarrow \phi_{3,0}(01) = 0, \theta_{3,0}(01) = \pi \Rightarrow \phi_{3,0}(01) = 0$. \\
    & \rot{$\ket{q_{3,1}(01)}$} & $\alpha_{3,1}(01) = \cos\frac{\theta_{3,1}(01)}{2}, \beta_{R_{3,1}}(01) = \cos\phi_{3,1}(01) \sin\frac{\theta_{3,1}(01)}{2}, \beta_{I_{3,1}}(01) = \sin\phi_{3,1}(01) \sin\frac{\theta_{3,1}(01)}{2},$ \\
    &                & $0 \leq \theta_{3,1}(01) \leq \pi, 0 \leq \phi_{3,1}(01) < 2\pi, \theta_{3,1}(01) = 0 \Rightarrow \phi_{3,1}(01) = 0, \theta_{3,1}(01) = \pi \Rightarrow \phi_{3,1}(01) = 0$. \\
    & \rot{$\ket{q_{3,2}(01)}$} & $\alpha_{3,2}(01) = \cos\frac{\theta_{3,2}(01)}{2}, \beta_{R_{3,2}}(01) = \cos\phi_{3,2}(01) \sin\frac{\theta_{3,2}(01)}{2}, \beta_{I_{3,2}}(01) = \sin\phi_{3,2}(01) \sin\frac{\theta_{3,2}(01)}{2},$ \\
    &                & $0 \leq \theta_{3,2}(01) \leq \pi, 0 \leq \phi_{3,2}(01) < 2\pi, \theta_{3,2}(01) = 0 \Rightarrow \phi_{3,2}(01) = 0, \theta_{3,2}(01) = \pi \Rightarrow \phi_{3,2}(01) = 0$. \\ \midrule
    \rot{$s_3(10)$} 
    & \rot{$\ket{q_{3,0}(10)}$} & $\alpha_{3,0}(10) = \cos\frac{\theta_{3,0}(10)}{2}, \beta_{R_{3,0}}(10) = \cos\phi_{3,0}(10) \sin\frac{\theta_{3,0}(10)}{2}, \beta_{I_{3,0}}(10) = \sin\phi_{3,0}(10) \sin\frac{\theta_{3,0}(10)}{2},$ \\
    &               & $0 \leq \theta_{3,0}(10) \leq \pi, 0 \leq \phi_{3,0}(10) < 2\pi, \theta_{3,0}(10) = 0 \Rightarrow \phi_{3,0}(10) = 0, \theta_{3,0}(10) = \pi \Rightarrow \phi_{3,0}(10) = 0$. \\
    & \rot{$\ket{q_{3,1}(10)}$} & $\alpha_{3,1}(10) = \cos\frac{\theta_{3,1}(10)}{2}, \beta_{R_{3,1}}(10) = \cos\phi_{3,1}(10) \sin\frac{\theta_{3,1}(10)}{2}, \beta_{I_{3,1}}(10) = \sin\phi_{3,1}(10) \sin\frac{\theta_{3,1}(10)}{2},$ \\
    &                & $0 \leq \theta_{3,1}(10) \leq \pi, 0 \leq \phi_{3,1}(10) < 2\pi, \theta_{3,1}(10) = 0 \Rightarrow \phi_{3,1}(10) = 0, \theta_{3,1}(10) = \pi \Rightarrow \phi_{3,1}(10) = 0$. \\
    & \rot{$\ket{q_{3,2}(10)}$} & $\alpha_{3,2}(10) = \cos\frac{\theta_{3,2}(10)}{2}, \beta_{R_{3,2}}(10) = \cos\phi_{3,2}(10) \sin\frac{\theta_{3,2}(10)}{2}, \beta_{I_{3,2}}(10) = \sin\phi_{3,2}(10) \sin\frac{\theta_{3,2}(10)}{2},$ \\
    &                & $0 \leq \theta_{3,2}(10) \leq \pi, 0 \leq \phi_{3,2}(10) < 2\pi, \theta_{3,2}(10) = 0 \Rightarrow \phi_{3,2}(10) = 0, \theta_{3,2}(10) = \pi \Rightarrow \phi_{3,2}(10) = 0$. \\ \midrule
    \rot{$s_3(11)$} 
    & \rot{$\ket{q_{3,0}(11)}$} & $\alpha_{3,0}(11) = \cos\frac{\theta_{3,0}(11)}{2}, \beta_{R_{3,0}}(11) = \cos\phi_{3,0}(11) \sin\frac{\theta_{3,0}(11)}{2}, \beta_{I_{3,0}}(11) = \sin\phi_{3,0}(11) \sin\frac{\theta_{3,0}(11)}{2},$ \\
    &               & $0 \leq \theta_{3,0}(11) \leq \pi, 0 \leq \phi_{3,0}(11) < 2\pi, \theta_{3,0}(11) = 0 \Rightarrow \phi_{3,0}(11) = 0, \theta_{3,0}(11) = \pi \Rightarrow \phi_{3,0}(11) = 0$. \\
    & \rot{$\ket{q_{3,1}(11)}$} & $\alpha_{3,1}(11) = \cos\frac{\theta_{3,1}(11)}{2}, \beta_{R_{3,1}}(11) = \cos\phi_{3,1}(11) \sin\frac{\theta_{3,1}(11)}{2}, \beta_{I_{3,1}}(11) = \sin\phi_{3,1}(11) \sin\frac{\theta_{3,1}(11)}{2},$ \\
    &                & $0 \leq \theta_{3,1}(11) \leq \pi, 0 \leq \phi_{3,1}(11) < 2\pi, \theta_{3,1}(11) = 0 \Rightarrow \phi_{3,1}(11) = 0, \theta_{3,1}(11) = \pi \Rightarrow \phi_{3,1}(11) = 0$. \\
    & \rot{$\ket{q_{3,2}(11)}$} & $\alpha_{3,2}(11) = \cos\frac{\theta_{3,2}(11)}{2}, \beta_{R_{3,2}}(11) = \cos\phi_{3,2}(11) \sin\frac{\theta_{3,2}(11)}{2}, \beta_{I_{3,2}}(11) = \sin\phi_{3,2}(11) \sin\frac{\theta_{3,2}(11)}{2},$ \\
    &                & $0 \leq \theta_{3,2}(11) \leq \pi, 0 \leq \phi_{3,2}(11) < 2\pi, \theta_{3,2}(11) = 0 \Rightarrow \phi_{3,2}(11) = 0, \theta_{3,2}(11) = \pi \Rightarrow \phi_{3,2}(11) = 0$. \\ \midrule
    \rot{$s_4(00)$} 
    & \rot{$\ket{q_{4,0}(00)}$} & $\alpha_{4,0}(00) = \cos\frac{\theta_{4,0}(00)}{2}, \beta_{R_{4,0}}(00) = \cos\phi_{4,0}(00) \sin\frac{\theta_{4,0}(00)}{2}, \beta_{I_{4,0}}(00) = \sin\phi_{4,0}(00) \sin\frac{\theta_{4,0}(00)}{2},$ \\
    &               & $0 \leq \theta_{4,0}(00) \leq \pi, 0 \leq \phi_{4,0}(00) < 2\pi, \theta_{4,0}(00) = 0 \Rightarrow \phi_{4,0}(00) = 0, \theta_{4,0}(00) = \pi \Rightarrow \phi_{4,0}(00) = 0$. \\
    & \rot{$\ket{q_{4,1}(00)}$} & $\alpha_{4,1}(00) = \cos\frac{\theta_{4,1}(00)}{2}, \beta_{R_{4,1}}(00) = \cos\phi_{4,1}(00) \sin\frac{\theta_{4,1}(00)}{2}, \beta_{I_{4,1}}(00) = \sin\phi_{4,1}(00) \sin\frac{\theta_{4,1}(00)}{2},$ \\
    &                & $0 \leq \theta_{4,1}(00) \leq \pi, 0 \leq \phi_{4,1}(00) < 2\pi, \theta_{4,1}(00) = 0 \Rightarrow \phi_{4,1}(00) = 0, \theta_{4,1}(00) = \pi \Rightarrow \phi_{4,1}(00) = 0$. \\
    & \rot{$\ket{q_{4,2}(00)}$} & $\alpha_{4,2}(00) = \cos\frac{\theta_{4,2}(00)}{2}, \beta_{R_{4,2}}(00) = \cos\phi_{4,2}(00) \sin\frac{\theta_{4,2}(00)}{2}, \beta_{I_{4,2}}(00) = \sin\phi_{4,2}(00) \sin\frac{\theta_{4,2}(00)}{2},$ \\
    &                & $0 \leq \theta_{4,2}(00) \leq \pi, 0 \leq \phi_{4,2}(00) < 2\pi, \theta_{4,2}(00) = 0 \Rightarrow \phi_{4,2}(00) = 0, \theta_{4,2}(00) = \pi \Rightarrow \phi_{4,2}(00) = 0$. \\ \midrule
    \rot{$s_4(01)$} 
    & \rot{$\ket{q_{4,0}(01)}$} & $\alpha_{4,0}(01) = \cos\frac{\theta_{4,0}(01)}{2}, \beta_{R_{4,0}}(01) = \cos\phi_{4,0}(01) \sin\frac{\theta_{4,0}(01)}{2}, \beta_{I_{4,0}}(01) = \sin\phi_{4,0}(01) \sin\frac{\theta_{4,0}(01)}{2},$ \\
    &               & $0 \leq \theta_{4,0}(01) \leq \pi, 0 \leq \phi_{4,0}(01) < 2\pi, \theta_{4,0}(01) = 0 \Rightarrow \phi_{4,0}(01) = 0, \theta_{4,0}(01) = \pi \Rightarrow \phi_{4,0}(01) = 0$. \\
    & \rot{$\ket{q_{4,1}(01)}$} & $\alpha_{4,1}(01) = \cos\frac{\theta_{4,1}(01)}{2}, \beta_{R_{4,1}}(01) = \cos\phi_{4,1}(01) \sin\frac{\theta_{4,1}(01)}{2}, \beta_{I_{4,1}}(01) = \sin\phi_{4,1}(01) \sin\frac{\theta_{4,1}(01)}{2},$ \\
    &                & $0 \leq \theta_{4,1}(01) \leq \pi, 0 \leq \phi_{4,1}(01) < 2\pi, \theta_{4,1}(01) = 0 \Rightarrow \phi_{4,1}(01) = 0, \theta_{4,1}(01) = \pi \Rightarrow \phi_{4,1}(01) = 0$. \\
    & \rot{$\ket{q_{4,2}(01)}$} & $\alpha_{4,2}(01) = \cos\frac{\theta_{4,2}(01)}{2}, \beta_{R_{4,2}}(01) = \cos\phi_{4,2}(01) \sin\frac{\theta_{4,2}(01)}{2}, \beta_{I_{4,2}}(01) = \sin\phi_{4,2}(01) \sin\frac{\theta_{4,2}(01)}{2},$ \\
    &                & $0 \leq \theta_{4,2}(01) \leq \pi, 0 \leq \phi_{4,2}(01) < 2\pi, \theta_{4,2}(01) = 0 \Rightarrow \phi_{4,2}(01) = 0, \theta_{4,2}(01) = \pi \Rightarrow \phi_{4,2}(01) = 0$. \\ \midrule
    \rot{$s_4(10)$} 
    & \rot{$\ket{q_{4,0}(10)}$} & $\alpha_{4,0}(10) = \cos\frac{\theta_{4,0}(10)}{2}, \beta_{R_{4,0}}(10) = \cos\phi_{4,0}(10) \sin\frac{\theta_{4,0}(10)}{2}, \beta_{I_{4,0}}(10) = \sin\phi_{4,0}(10) \sin\frac{\theta_{4,0}(10)}{2},$ \\
    &               & $0 \leq \theta_{4,0}(10) \leq \pi, 0 \leq \phi_{4,0}(10) < 2\pi, \theta_{4,0}(10) = 0 \Rightarrow \phi_{4,0}(10) = 0, \theta_{4,0}(10) = \pi \Rightarrow \phi_{4,0}(10) = 0$. \\
    & \rot{$\ket{q_{4,1}(10)}$} & $\alpha_{4,1}(10) = \cos\frac{\theta_{4,1}(10)}{2}, \beta_{R_{4,1}}(10) = \cos\phi_{4,1}(10) \sin\frac{\theta_{4,1}(10)}{2}, \beta_{I_{4,1}}(10) = \sin\phi_{4,1}(10) \sin\frac{\theta_{4,1}(10)}{2},$ \\
    &                & $0 \leq \theta_{4,1}(10) \leq \pi, 0 \leq \phi_{4,1}(10) < 2\pi, \theta_{4,1}(10) = 0 \Rightarrow \phi_{4,1}(10) = 0, \theta_{4,1}(10) = \pi \Rightarrow \phi_{4,1}(10) = 0$. \\
    & \rot{$\ket{q_{4,2}(10)}$} & $\alpha_{4,2}(10) = \cos\frac{\theta_{4,2}(10)}{2}, \beta_{R_{4,2}}(10) = \cos\phi_{4,2}(10) \sin\frac{\theta_{4,2}(10)}{2}, \beta_{I_{4,2}}(10) = \sin\phi_{4,2}(10) \sin\frac{\theta_{4,2}(10)}{2},$ \\
    &                & $0 \leq \theta_{4,2}(10) \leq \pi, 0 \leq \phi_{4,2}(10) < 2\pi, \theta_{4,2}(10) = 0 \Rightarrow \phi_{4,2}(10) = 0, \theta_{4,2}(10) = \pi \Rightarrow \phi_{4,2}(10) = 0$. \\ \midrule
    \rot{$s_4(11)$} 
    & \rot{$\ket{q_{4,0}(11)}$} & $\alpha_{4,0}(11) = \cos\frac{\theta_{4,0}(11)}{2}, \beta_{R_{4,0}}(11) = \cos\phi_{4,0}(11) \sin\frac{\theta_{4,0}(11)}{2}, \beta_{I_{4,0}}(11) = \sin\phi_{4,0}(11) \sin\frac{\theta_{4,0}(11)}{2},$ \\
    &               & $0 \leq \theta_{4,0}(11) \leq \pi, 0 \leq \phi_{4,0}(11) < 2\pi, \theta_{4,0}(11) = 0 \Rightarrow \phi_{4,0}(11) = 0, \theta_{4,0}(11) = \pi \Rightarrow \phi_{4,0}(11) = 0$. \\
    & \rot{$\ket{q_{4,1}(11)}$} & $\alpha_{4,1}(11) = \cos\frac{\theta_{4,1}(11)}{2}, \beta_{R_{4,1}}(11) = \cos\phi_{4,1}(11) \sin\frac{\theta_{4,1}(11)}{2}, \beta_{I_{4,1}}(11) = \sin\phi_{4,1}(11) \sin\frac{\theta_{4,1}(11)}{2},$ \\
    &                & $0 \leq \theta_{4,1}(11) \leq \pi, 0 \leq \phi_{4,1}(11) < 2\pi, \theta_{4,1}(11) = 0 \Rightarrow \phi_{4,1}(11) = 0, \theta_{4,1}(11) = \pi \Rightarrow \phi_{4,1}(11) = 0$. \\
    & \rot{$\ket{q_{4,2}(11)}$} & $\alpha_{4,2}(11) = \cos\frac{\theta_{4,2}(11)}{2}, \beta_{R_{4,2}}(11) = \cos\phi_{4,2}(11) \sin\frac{\theta_{4,2}(11)}{2}, \beta_{I_{4,2}}(11) = \sin\phi_{4,2}(11) \sin\frac{\theta_{4,2}(11)}{2},$ \\
    &                & $0 \leq \theta_{4,2}(11) \leq \pi, 0 \leq \phi_{4,2}(11) < 2\pi, \theta_{4,2}(11) = 0 \Rightarrow \phi_{4,2}(11) = 0, \theta_{4,2}(11) = \pi \Rightarrow \phi_{4,2}(11) = 0$. \\ \midrule
    \rot{$s_5(00)$} 
    & \rot{$\ket{q_{5,0}(00)}$} & $\alpha_{5,0}(00) = \cos\frac{\theta_{5,0}(00)}{2}, \beta_{R_{5,0}}(00) = \cos\phi_{5,0}(00) \sin\frac{\theta_{5,0}(00)}{2}, \beta_{I_{5,0}}(00) = \sin\phi_{5,0}(00) \sin\frac{\theta_{5,0}(00)}{2},$ \\
    &               & $0 \leq \theta_{5,0}(00) \leq \pi, 0 \leq \phi_{5,0}(00) < 2\pi, \theta_{5,0}(00) = 0 \Rightarrow \phi_{5,0}(00) = 0, \theta_{5,0}(00) = \pi \Rightarrow \phi_{5,0}(00) = 0$. \\
    & \rot{$\ket{q_{5,1}(00)}$} & $\alpha_{5,1}(00) = \cos\frac{\theta_{5,1}(00)}{2}, \beta_{R_{5,1}}(00) = \cos\phi_{5,1}(00) \sin\frac{\theta_{5,1}(00)}{2}, \beta_{I_{5,1}}(00) = \sin\phi_{5,1}(00) \sin\frac{\theta_{5,1}(00)}{2},$ \\
    &                & $0 \leq \theta_{5,1}(00) \leq \pi, 0 \leq \phi_{5,1}(00) < 2\pi, \theta_{5,1}(00) = 0 \Rightarrow \phi_{5,1}(00) = 0, \theta_{5,1}(00) = \pi \Rightarrow \phi_{5,1}(00) = 0$. \\
    & \rot{$\ket{q_{5,2}(00)}$} & $\alpha_{5,2}(00) = \cos\frac{\theta_{5,2}(00)}{2}, \beta_{R_{5,2}}(00) = \cos\phi_{5,2}(00) \sin\frac{\theta_{5,2}(00)}{2}, \beta_{I_{5,2}}(00) = \sin\phi_{5,2}(00) \sin\frac{\theta_{5,2}(00)}{2},$ \\
    &                & $0 \leq \theta_{5,2}(00) \leq \pi, 0 \leq \phi_{5,2}(00) < 2\pi, \theta_{5,2}(00) = 0 \Rightarrow \phi_{5,2}(00) = 0, \theta_{5,2}(00) = \pi \Rightarrow \phi_{5,2}(00) = 0$. \\ \midrule
    \rot{$s_4(01)$} 
    & \rot{$\ket{q_{5,0}(01)}$} & $\alpha_{5,0}(01) = \cos\frac{\theta_{5,0}(01)}{2}, \beta_{R_{5,0}}(01) = \cos\phi_{5,0}(01) \sin\frac{\theta_{5,0}(01)}{2}, \beta_{I_{5,0}}(01) = \sin\phi_{5,0}(01) \sin\frac{\theta_{5,0}(01)}{2},$ \\
    &               & $0 \leq \theta_{5,0}(01) \leq \pi, 0 \leq \phi_{5,0}(01) < 2\pi, \theta_{5,0}(01) = 0 \Rightarrow \phi_{5,0}(01) = 0, \theta_{5,0}(01) = \pi \Rightarrow \phi_{5,0}(01) = 0$. \\
    & \rot{$\ket{q_{5,1}(01)}$} & $\alpha_{5,1}(01) = \cos\frac{\theta_{5,1}(01)}{2}, \beta_{R_{5,1}}(01) = \cos\phi_{5,1}(01) \sin\frac{\theta_{5,1}(01)}{2}, \beta_{I_{5,1}}(01) = \sin\phi_{5,1}(01) \sin\frac{\theta_{5,1}(01)}{2},$ \\
    &                & $0 \leq \theta_{5,1}(01) \leq \pi, 0 \leq \phi_{5,1}(01) < 2\pi, \theta_{5,1}(01) = 0 \Rightarrow \phi_{5,1}(01) = 0, \theta_{5,1}(01) = \pi \Rightarrow \phi_{5,1}(01) = 0$. \\
    & \rot{$\ket{q_{5,2}(01)}$} & $\alpha_{5,2}(01) = \cos\frac{\theta_{5,2}(01)}{2}, \beta_{R_{5,2}}(01) = \cos\phi_{5,2}(01) \sin\frac{\theta_{5,2}(01)}{2}, \beta_{I_{5,2}}(01) = \sin\phi_{5,2}(01) \sin\frac{\theta_{5,2}(01)}{2},$ \\
    &                & $0 \leq \theta_{5,2}(01) \leq \pi, 0 \leq \phi_{5,2}(01) < 2\pi, \theta_{5,2}(01) = 0 \Rightarrow \phi_{5,2}(01) = 0, \theta_{5,2}(01) = \pi \Rightarrow \phi_{5,2}(01) = 0$. \\ \midrule
    \rot{$s_4(10)$} 
    & \rot{$\ket{q_{5,0}(10)}$} & $\alpha_{5,0}(10) = \cos\frac{\theta_{5,0}(10)}{2}, \beta_{R_{5,0}}(10) = \cos\phi_{5,0}(10) \sin\frac{\theta_{5,0}(10)}{2}, \beta_{I_{5,0}}(10) = \sin\phi_{5,0}(10) \sin\frac{\theta_{5,0}(10)}{2},$ \\
    &               & $0 \leq \theta_{5,0}(10) \leq \pi, 0 \leq \phi_{5,0}(10) < 2\pi, \theta_{5,0}(10) = 0 \Rightarrow \phi_{5,0}(10) = 0, \theta_{5,0}(10) = \pi \Rightarrow \phi_{5,0}(10) = 0$. \\
    & \rot{$\ket{q_{5,1}(10)}$} & $\alpha_{5,1}(10) = \cos\frac{\theta_{5,1}(10)}{2}, \beta_{R_{5,1}}(10) = \cos\phi_{5,1}(10) \sin\frac{\theta_{5,1}(10)}{2}, \beta_{I_{5,1}}(10) = \sin\phi_{5,1}(10) \sin\frac{\theta_{5,1}(10)}{2},$ \\
    &                & $0 \leq \theta_{5,1}(10) \leq \pi, 0 \leq \phi_{5,1}(10) < 2\pi, \theta_{5,1}(10) = 0 \Rightarrow \phi_{5,1}(10) = 0, \theta_{5,1}(10) = \pi \Rightarrow \phi_{5,1}(10) = 0$. \\
    & \rot{$\ket{q_{5,2}(10)}$} & $\alpha_{5,2}(10) = \cos\frac{\theta_{5,2}(10)}{2}, \beta_{R_{5,2}}(10) = \cos\phi_{5,2}(10) \sin\frac{\theta_{5,2}(10)}{2}, \beta_{I_{5,2}}(10) = \sin\phi_{5,2}(10) \sin\frac{\theta_{5,2}(10)}{2},$ \\
    &                & $0 \leq \theta_{5,2}(10) \leq \pi, 0 \leq \phi_{5,2}(10) < 2\pi, \theta_{5,2}(10) = 0 \Rightarrow \phi_{5,2}(10) = 0, \theta_{5,2}(10) = \pi \Rightarrow \phi_{5,2}(10) = 0$. \\ \midrule
    \rot{$s_4(11)$} 
    & \rot{$\ket{q_{5,0}(11)}$} & $\alpha_{5,0}(11) = \cos\frac{\theta_{5,0}(11)}{2}, \beta_{R_{5,0}}(11) = \cos\phi_{5,0}(11) \sin\frac{\theta_{5,0}(11)}{2}, \beta_{I_{5,0}}(11) = \sin\phi_{5,0}(11) \sin\frac{\theta_{5,0}(11)}{2},$ \\
    &               & $0 \leq \theta_{5,0}(11) \leq \pi, 0 \leq \phi_{5,0}(11) < 2\pi, \theta_{5,0}(11) = 0 \Rightarrow \phi_{5,0}(11) = 0, \theta_{5,0}(11) = \pi \Rightarrow \phi_{5,0}(11) = 0$. \\
    & \rot{$\ket{q_{5,1}(11)}$} & $\alpha_{5,1}(11) = \cos\frac{\theta_{5,1}(11)}{2}, \beta_{R_{5,1}}(11) = \cos\phi_{5,1}(11) \sin\frac{\theta_{5,1}(11)}{2}, \beta_{I_{5,1}}(11) = \sin\phi_{5,1}(11) \sin\frac{\theta_{5,1}(11)}{2},$ \\
    &                & $0 \leq \theta_{5,1}(11) \leq \pi, 0 \leq \phi_{5,1}(11) < 2\pi, \theta_{5,1}(11) = 0 \Rightarrow \phi_{5,1}(11) = 0, \theta_{5,1}(11) = \pi \Rightarrow \phi_{5,1}(11) = 0$. \\
    & \rot{$\ket{q_{5,2}(11)}$} & $\alpha_{5,2}(11) = \cos\frac{\theta_{5,2}(11)}{2}, \beta_{R_{5,2}}(11) = \cos\phi_{5,2}(11) \sin\frac{\theta_{5,2}(11)}{2}, \beta_{I_{5,2}}(11) = \sin\phi_{5,2}(11) \sin\frac{\theta_{5,2}(11)}{2},$ \\
    &                & $0 \leq \theta_{5,2}(11) \leq \pi, 0 \leq \phi_{5,2}(11) < 2\pi, \theta_{5,2}(11) = 0 \Rightarrow \phi_{5,2}(11) = 0, \theta_{5,2}(11) = \pi \Rightarrow \phi_{5,2}(11) = 0$. 
    \\ \bottomrule
    \caption{Full SMT formula of teleportation: qubit constraints.}
    \label{tab:tp-smt-verbose-qubit-constraints}
\end{tabularx}

\normalsize

\begin{tabularx}{1.05\linewidth}{@{}lXX@{}}\toprule
    \textbf{State vector}       & \textbf{Variables} \\ \midrule
    $s_0$              & $s_{{0,0}_R}, s_{{0,0}_I}, s_{{0,1}_R}, s_{{0,1}_I}, s_{{0,2}_R}, s_{{0,2}_I}, s_{{0,3}_R}, s_{{0,3}_I}$ \\ \midrule
    \textbf{Initial valuation}  & \textbf{Assertions} \\ \midrule
    $\ket{q_1}, \ket{q_2}$ & $s_{{0,0}_R} = \alpha_{0,1}\alpha_{0,2}$,          & $s_{{0,0}_I} = 0$, \\
                           & $s_{{0,1}_R} = \alpha_{0,1}\beta_{R_{0,2}}$,       & $ s_{{0,1}_I} = \alpha_{0,1}\beta_{I_{0,2}}$, \\
                           & $s_{{0,2}_R} = \alpha_{0,2}\beta_{R_{0,1}}$,       & $ s_{{0,2}_I} = \alpha_{0,2}\beta_{I_{0,1}}$, \\
                           & $s_{{0,3}_R} = \beta_{R_{0,1}}\beta_{R_{0,2}} - \beta_{I_{0,1}}\beta_{I_{0,2}}$,   
                           & $s_{{0,3}_I} = \beta_{I_{0,1}}\beta_{R_{0,2}} + \beta_{R_{0,1}}\beta_{I_{0,2}}$, \\
                           & $s_{{0,0}_R} = 1, s_{{0,0}_I} = 0$, \\ 
                           & $s_{{0,1}_R} = 0, s_{{0,1}_I} = 0$, \\
                           & $s_{{0,2}_R} = 0, s_{{0,2}_I} = 0$, \\ 
                           & $s_{{0,3}_R} = 1, s_{{0,3}_I} = 0$.
    \\ \bottomrule
    \caption{Full SMT formula of teleportation: initial valuation.}
    \label{tab:tp-smt-verbose-initial}
\end{tabularx}

\small

\begin{tabularx}{1.05\linewidth}{@{}lcX@{}}\toprule
    \multicolumn{2}{l}{\textbf{State vector}} & \textbf{Variables} \\ \midrule
    \multicolumn{2}{l}{$s_1$}        & $s_{{1,0}_R}, s_{{1,0}_I}, s_{{1,1}_R}, s_{{1,1}_I}, s_{{1,2}_R}, s_{{1,2}_I}, s_{{1,3}_R}, s_{{1,3}_I}$ \\ \midrule
    \textbf{State} & \textbf{Operation}       & \textbf{Assertions} \\ \midrule
    $s_1$ & \rotmulti{6}{$U_{CX}(\ket{q_0}\!,\ket{q_1})$}     
       & $s_{{1,0}_R} = \alpha_{1,0}\alpha_{1,1}$,           $s_{{1,0}_I} = 0$, \\
    &  & $s_{{1,0}_R} = \alpha_{1,0}\beta_{R_{1,1}}$,        $ s_{{1,0}_I} = \alpha_{1,0}\beta_{I_{1,1}}$, \\
    &  & $s_{{1,1}_R} = \alpha_{1,1}\beta_{R_{1,0}}$,        $ s_{{1,1}_I} = \alpha_{1,1}\beta_{I_{1,0}}$, \\
    &  & $s_{{1,3}_R} = \beta_{R_{1,0}}\beta_{R_{1,1}} - \beta_{I_{1,0}}\beta_{I_{1,1}}$,   
     $s_{{1,3}_I} = \beta_{I_{1,0}}\beta_{R_{1,1}} + \beta_{R_{1,0}}\beta_{I_{1,1}}$, \\
    &  & $s_{{1,0}_R} = s_{{0,0}_R}, s_{{1,0}_I} = s_{{0,0}_I}$, \\ 
    &  & $s_{{1,1}_R} = s_{{0,2}_R}, s_{{1,1}_I} = s_{{0,2}_I}$, \\
    &  & $s_{{1,2}_R} = s_{{0,1}_R}, s_{{1,2}_I} = s_{{0,1}_I}$, \\ 
    &  & $s_{{1,3}_R} = s_{{0,3}_R}, s_{{1,3}_I} = s_{{0,3}_I}$, \\ 
    &  & $\alpha_{1,2} = \alpha_{0,2}, \beta_{R_{1,2}} = \beta_{R_{0,2}}, \beta_{I_{1,2}} = \beta_{I_{0,2}}$ \\ \midrule
    $s_2$ & \rotmulti{3}{$H(\ket{q_0})$}                   & $\alpha_{2,0} = (\alpha_{1,0} + \beta_{R_{1,0}})/\sqrt{2}$, 
                                                $\beta_{R_{2,0}} = (\alpha_{1,0} - \beta_{R_{1,0}})/\sqrt{2}$ \\
                                            & & $\alpha_{2,1} = \alpha_{1,1}, \beta_{R_{2,1}} = \beta_{R_{1,1}}, \beta_{I_{2,1}} = \beta_{I_{1,1}},$ \\
                                            & & $\alpha_{2,2} = \alpha_{1,2}, \beta_{R_{2,2}} = \beta_{R_{1,2}}, \beta_{I_{2,2}} = \beta_{I_{1,2}}$ \\ 
    \midrule
    $s_3(00)$ & \rotmulti{9}{$M(\ket{q_0}), M(\ket{q_1})$}  & $\alpha_{3,0}(00) = 1, \beta_{R_{3,0}}(00) = 0, \beta_{I_{3,0}}(00) = 0$, \\
                                            & & $\alpha_{3,1}(00) = 1, \beta_{R_{3,1}}(00) = 0, \beta_{I_{3,1}}(00) = 0$, \\
                                            & & $\alpha_{3,2}(00) = \alpha_{2,2}, \beta_{R_{3,2}}(00) = \beta_{R_{2,2}}, \beta_{I_{3,2}}(00) = \beta_{I_{2,2}}$, \\
    $s_3(01)$                               & & $\alpha_{3,0}(01) = 1, \beta_{R_{3,0}}(01) = 0, \beta_{I_{3,0}}(01) = 0$, \\
                                            & & $\alpha_{3,1}(01) = 0, \beta_{R_{3,1}}(01) = 1, \beta_{I_{3,1}}(01) = 0$, \\
                                            & & $\alpha_{3,2}(01) = \alpha_{2,2}, \beta_{R_{3,2}}(01) = \beta_{R_{2,2}}, \beta_{I_{3,2}}(01) = \beta_{I_{2,2}}$, \\
    $s_3(10)$                               & & $\alpha_{3,0}(10) = 0, \beta_{R_{3,0}}(10) = 1, \beta_{I_{3,0}}(10) = 0$, \\
                                            & & $\alpha_{3,1}(10) = 1, \beta_{R_{3,1}}(10) = 0, \beta_{I_{3,1}}(10) = 0$, \\   
                                            & & $\alpha_{3,2}(10) = \alpha_{2,2}, \beta_{R_{3,2}}(10) = \beta_{R_{2,2}}, \beta_{I_{3,2}}(10) = \beta_{I_{2,2}}$, \\
    $s_3(11)$                               & & $\alpha_{3,0}(11) = 0, \beta_{R_{3,0}}(11) = 1, \beta_{I_{3,0}}(11) = 0$, \\
                                            & & $\alpha_{3,1}(11) = 0, \beta_{R_{3,1}}(11) = 1, \beta_{I_{3,1}}(11) = 0$, \\
                                            & & $\alpha_{3,2}(11) = \alpha_{2,2}, \beta_{R_{3,2}}(11) = \beta_{R_{2,2}}, \beta_{I_{3,2}}(11) = \beta_{I_{2,2}}$, \\
    \midrule
    $s_4(00)$ & \rotmulti{9}{$U_{CX}(\ket{q_1},\ket{q_2})$} & $\alpha_{4,0}(00) = \alpha_{3,0}(00), \beta_{R_{4,0}}(00) 
                                            \beta_{R_{3,0}}(00), \beta_{I_{4,0}}(00) = \beta_{I_{3,0}}(00)$, \\
                                            & & $\alpha_{4,1}(00) = \alpha_{3,1}(00), \beta_{R_{4,1}}(00) = \beta_{R_{3,1}}(00), \beta_{I_{4,1}}(00) = \beta_{I_{3,1}}(00)$, \\
                                            & & $\alpha_{4,2}(00) = \alpha_{3,2}(00), \beta_{R_{4,2}}(00) = \beta_{R_{3,2}}(00), \beta_{I_{4,2}}(00) = \beta_{I_{3,2}}(00)$, \\
    $s_4(01)$                               & & $\alpha_{4,0}(01) = \alpha_{3,0}(01), \beta_{R_{4,0}}(01) 
                                            \beta_{R_{3,0}}(01), \beta_{I_{4,0}}(01) = \beta_{I_{3,0}}(01)$, \\
                                            & & $\alpha_{4,1}(01) = \alpha_{3,1}(01), \beta_{R_{4,1}}(01) = \beta_{R_{3,1}}(01), \beta_{I_{4,1}}(01) = \beta_{I_{3,1}}(01)$, \\
                                            & & $\alpha_{4,2}(01) = \alpha_{3,2}(01), \beta_{R_{4,2}}(01) = \beta_{R_{3,2}}(01), \beta_{I_{4,2}}(01) = \beta_{I_{3,2}}(01)$, \\
    $s_4(10)$                               & & $\alpha_{4,0}(10) = \alpha_{3,0}(10), \beta_{R_{4,0}}(10) 
                                            \beta_{R_{3,0}}(10), \beta_{I_{4,0}}(10) = \beta_{I_{3,0}}(10)$, \\
                                            & & $\alpha_{4,1}(10) = \alpha_{3,1}(10), \beta_{R_{4,1}}(10) = \beta_{R_{3,1}}(10), \beta_{I_{4,1}}(10) = \beta_{I_{3,1}}(10)$, \\
                                            & & $\alpha_{4,2}(10) = \beta{3,2}(10), \beta_{R_{4,2}}(10) = \alpha_{R_{3,2}}(10), \beta_{I_{4,2}}(10) = \beta_{I_{3,2}}(10)$, 
                                            \\
    $s_4(11)$                               & & $\alpha_{4,0}(11) = \alpha_{3,0}(11), \beta_{R_{4,0}}(11) 
                                            \beta_{R_{3,0}}(11), \beta_{I_{4,0}}(11) = \beta_{I_{3,0}}(11)$, \\
                                            & & $\alpha_{4,1}(11) = \alpha_{3,1}(11), \beta_{R_{4,1}}(11) = \beta_{R_{3,1}}(11), \beta_{I_{4,1}}(11) = \beta_{I_{3,1}}(11)$, \\
                                            & & $\alpha_{4,2}(11) = \beta{3,2}(11), \beta_{R_{4,2}}(11) = \alpha_{R_{3,2}}(11), \beta_{I_{4,2}}(11) = \beta_{I_{3,2}}(11)$, 
                                            \\
    \midrule
    $s_5(00)$ &                               & $\alpha_{5,0}(00) = \alpha_{4,0}(00), \beta_{R_{5,0}}(00) 
                                            \beta_{R_{4,0}}(00), \beta_{I_{5,0}}(00) = \beta_{I_{4,0}}(00)$, \\
                                            & & $\alpha_{5,1}(00) = \alpha_{4,1}(00), \beta_{R_{5,1}}(00) = \beta_{R_{4,1}}(00), \beta_{I_{5,1}}(00) = \beta_{I_{4,1}}(00)$, \\
              & \rotmulti{9}{$U_{CZ}(\ket{q_0},\ket{q_2})$} & $\alpha_{5,2}(00) = \alpha_{4,2}(00), \beta_{R_{5,2}}(00) = \beta_{R_{4,2}}(00), \beta_{I_{5,2}}(00) = \beta_{I_{4,2}}(00)$, \\
    $s_5(01)$                               & & $\alpha_{5,0}(01) = \alpha_{4,0}(01), \beta_{R_{5,0}}(01) 
                                            \beta_{R_{4,0}}(01), \beta_{I_{5,0}}(01) = \beta_{I_{4,0}}(01)$, \\
                                            & & $\alpha_{5,1}(01) = \alpha_{4,1}(01), \beta_{R_{5,1}}(01) = \beta_{R_{4,1}}(01), \beta_{I_{5,1}}(01) = \beta_{I_{4,1}}(01)$, \\
                                            & & $\alpha_{5,2}(01) = \alpha_{4,2}(01), \beta_{R_{5,2}}(01) = \beta_{R_{4,2}}(01), \beta_{I_{5,2}}(01) = \beta_{I_{4,2}}(01)$, \\
    $s_5(10)$                               & & $\alpha_{5,0}(10) = \alpha_{4,0}(10), \beta_{R_{5,0}}(10) 
                                            \beta_{R_{4,0}}(10), \beta_{I_{5,0}}(10) = \beta_{I_{4,0}}(10)$, \\
                                            & & $\alpha_{5,1}(10) = \alpha_{4,1}(10), \beta_{R_{5,1}}(10) = \beta_{R_{4,1}}(10), \beta_{I_{5,1}}(10) = \beta_{I_{4,1}}(10)$, \\
                                            & & $\alpha_{5,2}(10) = \alpha{3,2}(10), \beta_{R_{5,2}}(10) = -\beta_{R_{4,2}}(10), \beta_{I_{5,2}}(10) = -\beta_{I_{4,2}}(10)$, 
                                            \\
    $s_5(11)$                               & & $\alpha_{5,0}(11) = \alpha_{4,0}(11), \beta_{R_{5,0}}(11) 
                                            \beta_{R_{4,0}}(11), \beta_{I_{5,0}}(11) = \beta_{I_{4,0}}(11)$, \\
                                            & & $\alpha_{5,1}(11) = \alpha_{4,1}(11), \beta_{R_{5,1}}(11) = \beta_{R_{4,1}}(11), \beta_{I_{5,1}}(11) = \beta_{I_{4,1}}(11)$, \\
                                            & & $\alpha_{5,2}(11) = \alpha{3,2}(11), \beta_{R_{5,2}}(11) = -\beta{R_{4,2}}(11), \beta_{I_{5,2}}(11) = -\beta_{I_{4,2}}(11)$, 
                                            \\ 
    \bottomrule
    \caption{Full SMT formula of teleportation: operations.}
    \label{tab:tp-smt-verbose-operations}
\end{tabularx}

\normalsize

\begin{tabularx}{1.05\linewidth}{@{}X@{}}\toprule
    \textbf{Assertions} \\ \midrule
    $\alpha_{5,2}(00) = \alpha_{0,0}, \quad \beta_{R_{5,2}}(00) = \beta_{R_{0,0}}, \quad \beta_{I_{5,2}}(00) = \beta_{I_{0,0}}$, \\
    $\alpha_{5,2}(01) = \alpha_{0,0}, \quad \beta_{R_{5,2}}(01) = \beta_{R_{0,0}}, \quad \beta_{I_{5,2}}(01) = \beta_{I_{0,0}}$, \\
    $\alpha_{5,2}(10) = \alpha_{0,0}, \quad \beta_{R_{5,2}}(10) = \beta_{R_{0,0}}, \quad \beta_{I_{5,2}}(10) = \beta_{I_{0,0}}$, \\
    $\alpha_{5,2}(11) = \alpha_{0,0}, \quad \beta_{R_{5,2}}(11) = \beta_{R_{0,0}}, \quad \beta_{I_{5,2}}(11) = \beta_{I_{0,0}}$. 
    \\ \bottomrule
    \caption{Full SMT formula of teleportation: specification.}
    \label{tab:tp-smt-verbose-specification}
\end{tabularx}

\end{document}